\documentclass[10pt,technote,onecolumn]{article}

\usepackage{setspace}
\usepackage{mathrsfs}
\usepackage{graphicx}
\usepackage{amssymb,amsmath}
\usepackage{eufrak}
\usepackage{color}
\usepackage{cite}

\DeclareMathAlphabet{\mathpzc}{OT1}{pzc}{m}{it}

\newcommand{\last}{\operatornamewithlimits{last}}
\newcommand{\REV}{\operatornamewithlimits{rev}}
\newcommand{\tog}{\operatornamewithlimits{tog}}
\newcommand{\dia}{\operatornamewithlimits{dia}}

\newtheorem{theorem}{Theorem}
\newtheorem{lemma}[theorem]{Lemma}

\newtheorem{corollary}[theorem]{Corollary}

\newenvironment{proof}[1][Proof]{\begin{trivlist}
\item[\hskip \labelsep {\bfseries #1}]}{\end{trivlist}}

\ifodd 1
\newcommand{\rev}[1]{{\color{blue}#1}} 
\newcommand{\com}[1]{\textbf{\color{red} (COMMENT: #1) }} 

\newcommand{\response}[1]{\textbf{\color{green} (RESPONSE: #1)}}

\else
\newcommand{\rev}[1]{#1}

\newcommand{\com}[1]{}
\newcommand{\response}[1]{}
\fi

\title{Complex Networks from Simple Rewrite Systems}

\author{Richard Southwell$^*$, Jianwei Huang$^*$, Chris Cannings$^@$
\ \\
$*$Information Engineering Department, The Chinese University of Hong Kong,\\
 Shatin, New Territories, Hong Kong,\\
richardsouthwell254@gmail.com \\
jianweihuang@gmail.com \\
\\
$@$School of Mathematics and Statistics, University of Sheffield,\\
Sheffield, S3 7RH, UK,\\
c.cannings@sheffield.ac.uk\\
}

\begin{document}
  \maketitle

  \begin{abstract}
Complex networks are all around us, and they can be generated by simple mechanisms. Understanding what kinds of networks can be produced by following simple rules is therefore of great importance. We investigate this issue by studying the dynamics of extremely simple systems where are `writer' moves around a network, and modifies it in a way that depends upon the writer's surroundings. Each vertex in the network has three edges incident upon it, which are colored red, blue and green. This edge coloring is done to provide a way for the writer to orient its movement. We explore the dynamics of a space of 3888 of these `colored trinet automata' systems. We find a large variety of behaviour, ranging from the very simple to the very complex. We also discover simple rules that generate forms which are remarkably similar to a wide range of natural objects. We study our systems using simulations (with appropriate visualization techniques) and analyze selected rules mathematically. We arrive at an empirical classification scheme which reveals a lot about the kinds of dynamics and networks that can be generated by these systems.
\end{abstract}

\section{Introduction}
Many advances in complex systems research have come from identifying simple deterministic systems which can generate complex behavior. For example, cellular automata \cite{wolfca} are used to model many systems in physics and biology, string rewrite \cite{lind} systems are used to model plant development, and chaotic differential equations \cite{ter} are used in many areas. The most popular models of complex network growth are probabilistic \cite{small,barabasi}. These models have yielded many insights, but they do not offer an explanation as to \emph{where} the complexity in networks comes from. The complexity in probabilistic models comes from the randomizing mechanisms they are based on.

In this paper we study a class of simple deterministic network growth models. Our systems can produce a wide variety of behaviour, ranging from the very simple to the very complex. They offer insights into how complex structures can be generated, because each aspect of their behaviour can be traced back to a deterministic cause. The simple nature of our models allows us to quickly generate exact pictures of their dynamics. This gives us the ability to easily explore the behaviour of large numbers of systems. By taking an unbiased look at the dynamics of many simple rules we see what types of network growth are easily generated by simple computational processes.

In Turing machines \cite{turing} a writer moves along a one dimensional tape, and rewrites symbols using local rules. Our \emph{network automata} systems are like a generalization of this idea -where the writer moves around a network, and rewrites it on a local level. The way the writer moves and modifies its local structure is determined by its surroundings. The network the writer runs on is a \emph{trinet} (i.e., each vertex has three connections). Trinets are also known as cubic graphs and three-regular graphs. There are many natural examples of trinets, such as two dimensional foams and polygonal networks formed by biological cells and cracks in the earth. We focus on trinets because they are easy to manipulate and expressive. Any network can be represented as a trinet by replacing vertices that have more than three connections with ``roundabout-like'' circles of vertices with three connections \cite{bol1}.

One of our goals is to find simple deterministic rules that generate complex dynamics. An obstacle to this goal is the inability of deterministic rules to break symmetries. For example, if the writer is at one of the corners of a cube then it has no deterministic way to select which of its ``similar looking'' neighbors it should move to next.  We circumvent this obstacle by supposing the edges of our trinet are colored red, blue and green in such a way that touching edges have different colors (such an edge coloring is also known as a Tait coloring \cite{tait}). This allows us to specify how the writer should move by stating which colored edge it should move along in different situations\footnote{We could remove the reliance of edge color in our systems, by replacing the red, blue and green edges with different uncolored structures, and modifying our rewrite rules accordingly. Although this would make the systems look more complicated.}.

Our main goal is to understand about what kinds of behaviour can be generated by colored trinet automata. To achieve this, we do a thorough exploration of the dynamics of a space of simple rules. Using pictures and analytic techniques we sort these rules into three classes (fixed points, repetitive growth and elaborate growth) according to their behaviour. We discuss these classes in sections \ref{class 1},\ref{class 2} and \ref{class 3}. In section \ref{other} we exhibit more general rules, which can produce other exotic behaviour such as persistent complex behaviour, periodically changing networks and production of hyperbolic space.
In section \ref{other} we also present many rules that can produce networks that look remarkably similar to real world objects.

We have tried to keep the main text as free from equations as possible. This was done to increase readership, and because complex phenomena are often better explained using pictures. We also studied our systems mathematically, and the appendix contains proofs which will give the mathematically inclined reader significant insights into some of our system's dynamics.


\subsection{Related Literature}\label{lit}

Many interesting deterministic network growth models have been previously considered. These include the growth models discussed in \cite{ticktock} and the biologically inspired models considered in \cite{rep1,rep2,rep3,rep4}, which can produce rich and complicated self replicating structures, despite extremely simple rules. Self replication in adaptive network models was also considered in \cite{jap}. In each of these cases, every part of the network gets updated in parallel. In our systems, by contrast, only a small piece of the network gets updated at any time. This makes our systems easier to implement because there is no need for distant parts of the network to have synchronized clocks. It is true that many real world networks grow in parallel, however evolving our systems for many time steps often leads to structural modifications that are equivalent to global/parrelel rewrite operations (see subsections \ref{cycsimpsec} and \ref{bouncesec}).

 Our work was inspired by the work of Tommaso Bolognesi \cite{bol1}, who studies the dynamics of planar trinet  automata. Like our models, these systems involve a writer which moves around trinet, making structural modifications as it goes. Unlike our systems, the models considered in \cite{bol1} circumvents the symmetry breaking problem by supposing the trinets are embedded in the plane. The way the writer behaves in these systems is governed by the structure of the faces formed by the planar embedding.  Although these models are fascinating, the fact that the dynamics depend on how the network is embedded makes it difficult to get a full picture of what is going on at any given time. Our approach (of considered edge-colored trinets, rather than planar embedded trinets) allows us to easily picture the complete dynamics (and rules) behind the systems we consider.

 The experimental approach we use was pioneered by Stephen Wolfram \cite{nks}, who uses simulations to reveal a vast array of simple programs that can generate complex dynamics.  Wolfram finds many network growth models that can generate complex dynamics and explores the idea that the physical universe could be generated by a simple network growth model. The idea that the universe can be generated in exact detail, by a simple program, is one of the most interesting conjectures from digital physics \cite{dig, cause, alex}. Wolfram suggests that the correct model could be found by doing an automated search of the simplest possibilities \cite{blog}. In order to do such \emph{universe hunting} one needs intuition about the kinds of things that simple adaptive network mechanisms can do. One of our aims is to provide some of this intuition.

\subsection{How Our Systems Work}\label{work}

A colored trinet automata is a dynamical system where a writer moves around the vertices of an edge-colored trinet; applying rewrite rules as it goes. Every time step some rewrite operation is applied about the writer's current location and then the writer moves. We start by considering only extremely simple rewrite operations, where the writer's current vertex is either replaced with a triangle or left unaltered. The action the writer takes on a given time step is determined by the colors of the edges interlinking its neighbours.

\begin{figure}
\centerline{\includegraphics[scale=0.6]{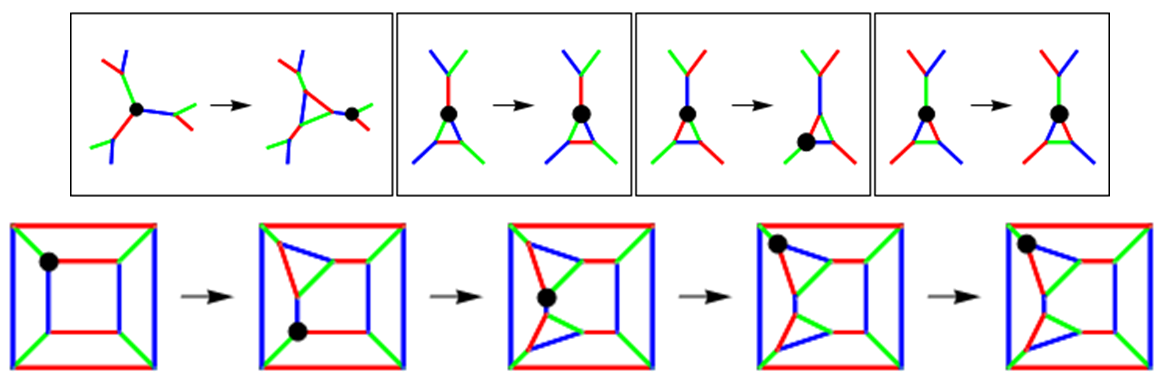}}
\vspace*{8pt}
\caption{An example colored trinet automata. At the top we show the rules of the system by indicating how the writer (the black vertex) moves and modifies the network in response to its surroundings. Underneath we show the system evolving for four updates, starting from the cube. The system reaches a fixed point after three updates.}
\label{firstfixed}
\end{figure}
The rules behind a colored trinet automata specify how the writer should move, and modify the network in response to its surroundings  (i.e., the colors of the edges interlinking the writer's neighbors). We picture the rules at the top of our figures, and the writer is represented by a black vertex. For example, the rule for the system shown in Figure \ref{firstfixed} can be described as follows;
\begin{enumerate}
\item If there are no edges linking the writer's neighbours then the writer then moves along a blue edge, and the writer's previous location is replaced with a triangle.
\item If there is exactly one edge linking the writer's neighbours, which is red, then take no action.
\item If there is exactly one edge linking the writer's neighbours, which is blue, then the writer along a red edge and the network is left unaltered.
\item If there is exactly one edge linking the writer's neighbours, which is green, then take no action.
\end{enumerate}
Notice how these instructions correspond to the pictures in the four boxes at the top of Figure \ref{firstfixed}. In general, a rule is just a specification of which structural modifications and movements the writer should perform in response to the four types of surroundings (depending on colors of edges interlinking the writer's neighbors) that the writer can have when there is no more than one edge linking its neighbors. Our rules do not specify which actions the writer should take in other situations, where there are two or more edges linking its neighbors. In these cases we suppose that no action is taken. Actually, the rules we consider never generate vertices with two or more linked neighbours and so these situations never occur anyway.



\subsection{The Rule Space}

In each of our investigations we used the cube (shown at the bottom left of Figure \ref{firstfixed}) as our initial condition. We studied the dynamics of the space of $12 \times 18 \times 18 = 3888$ rules depicted in Figure \ref{rulespace}.

\begin{figure}
\centerline{\includegraphics[scale=0.6]{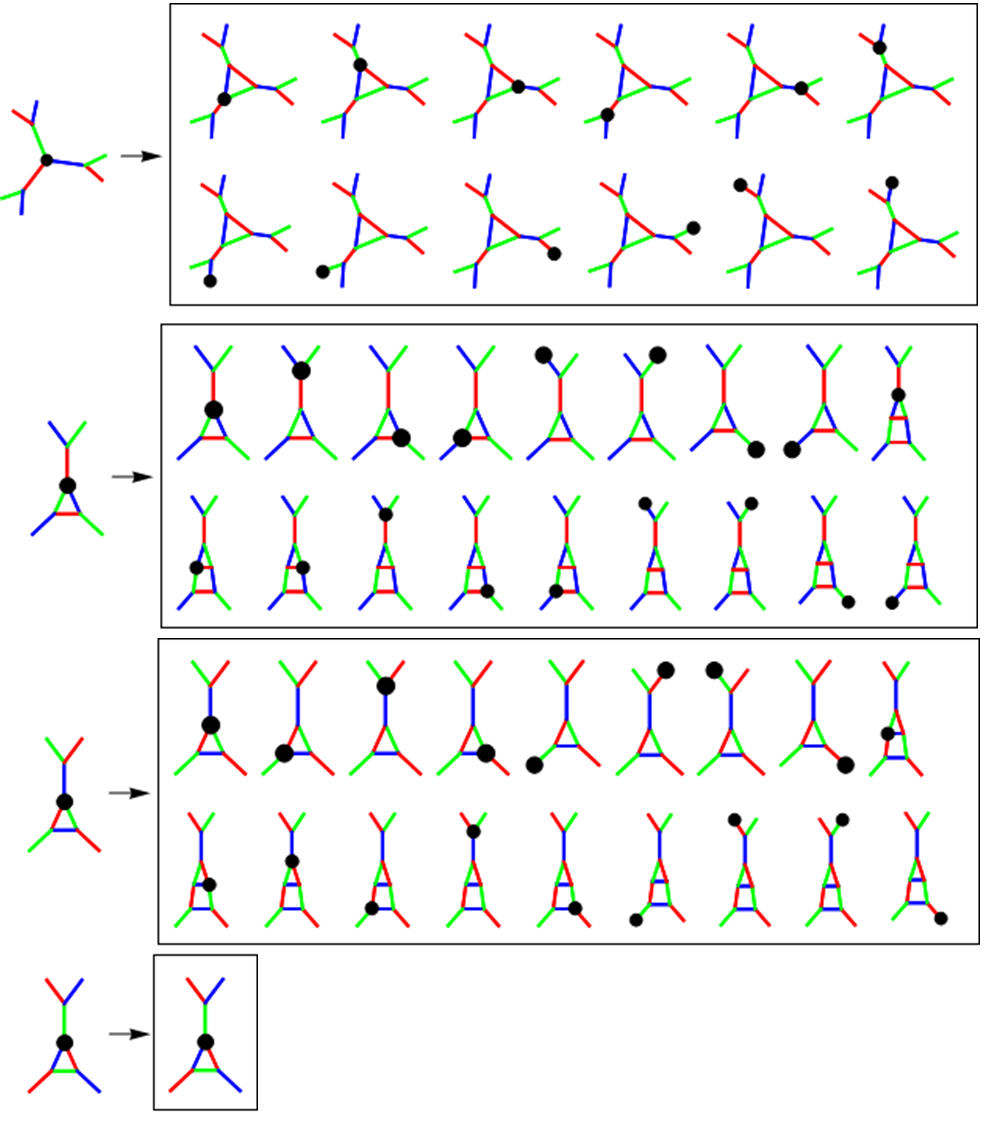}}
\vspace*{8pt}
\caption{An illustration of the space of rules we explored. A rule is specified by associating each of the four different types of surroundings with one of the images to its right.}
\label{rulespace}
\end{figure}

These are all the rules which involve possible triangle-replacement, and movement along two or less links, with two exceptions. Our first exception is that we always replace a vertex with no interlinked neighbours with a triangle.  We make this exception so that each of our rules is guaranteed to change our initial \emph{cube} network. Our second exception is that we never take action when there is a green edge linking the writer's neighbors. We make this exception to reduce the number of rules we must consider. If we drop this second constraint we find new kinds of behaviour, but the rule space becomes too large to thoroughly explore.

By severely limiting the kinds of operations our trinet automata can perform, we reduce the rule space to a manageable size and insure that it contains minimal systems that generate certain types of behaviour. We sort our $3888$ rules into three classes (fixed points, repetitive growth and elaborate growth) according to the long term dynamics they generate, starting from the cube. The following three sections are devoted to describing the behaviour of the rules in these three classes.

\section{Fixed Points}\label{class 1}

We say a system has reached a fixed point, where there comes a time after which the network no longer changes. $2918$ (about $75$ percent) of the our rules  eventually reach a fixed point. Our example system from Figure \ref{firstfixed} reaches a fixed point after three updates.
One reason such a large number of rules end up at a fixed point is that our rule space is such that whenever a green edge links a pair of the writer's neighbors a system becomes fixed. Indeed $2562$ of our rules halt their evolution because of this effect. These include the rule which grows the largest fixed structure, shown in Figure \ref{bigfixed}.

\begin{figure}
\centerline{\includegraphics[scale=0.7]{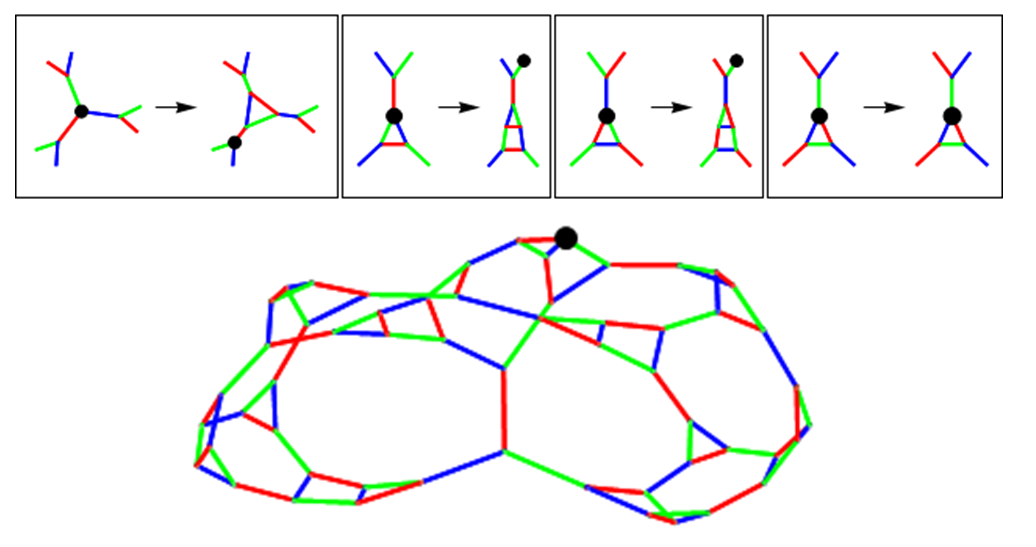}}
\vspace*{8pt}
\caption{Of all our rules, this one generates the largest static network (with $56$ vertices). It takes $26$ updates for the cube to change into this network.}
\label{bigfixed}
\end{figure}

In addition to the fixed points where the writer halts, there are \emph{dynamic-writer} fixed point where the writer continues to move forever, even after the network has become static. $162$ of the rules in class $1$ evolve into dynamic-writer fixed points. In $130$ of these rules the writer ends up doing a period two orbit (see Figure \ref{per2}), in the others the writer ends up doing a period four orbit.

\begin{figure}
\centerline{\includegraphics[scale=0.5]{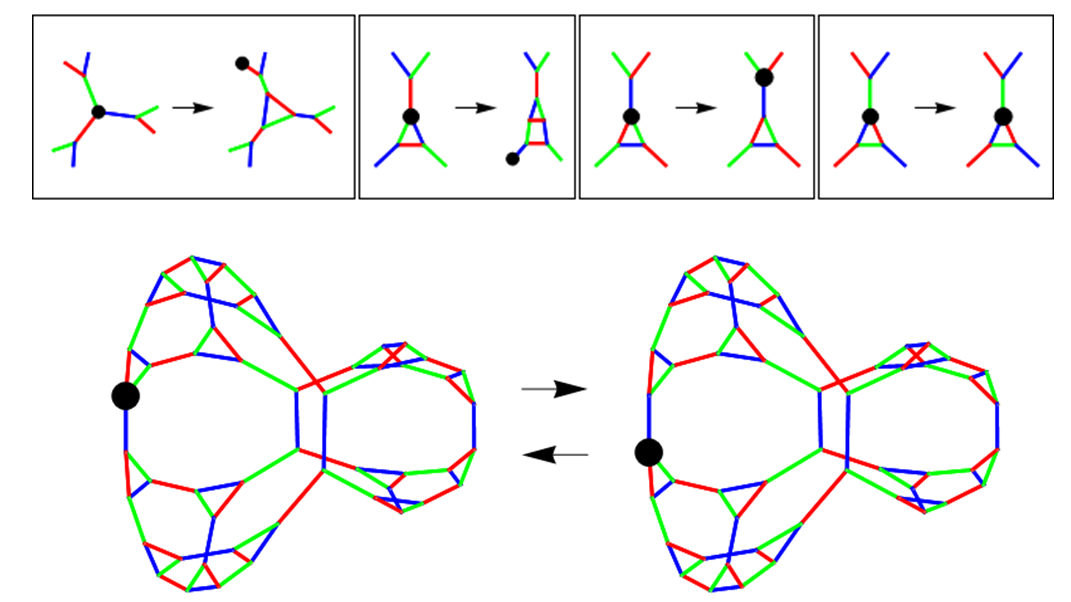}}
\vspace*{8pt}
\caption{Of all our rules, this one takes the longest ($34$ updates) to reach a static network (which has $52$ vertices). The writer continues to move in the period two orbit shown after this network has been generated.}
\label{per2}
\end{figure}

\section{Repetitive Growth}\label{class 2}

The next most complex type of behaviour observed is repetitive growth. This is characterized by the feature that the structure continues to grow, while the writer is ``trapped'' within a particular region of the network -with the form of its surroundings changing periodically. Repetitive growth occurs because writer's surroundings induce it to generate more structure around itself of the same form. $840$ of our rules eventually generate repetitive growth (see Figure \ref{repsimp}).

\begin{figure}
\centerline{\includegraphics[scale=0.7]{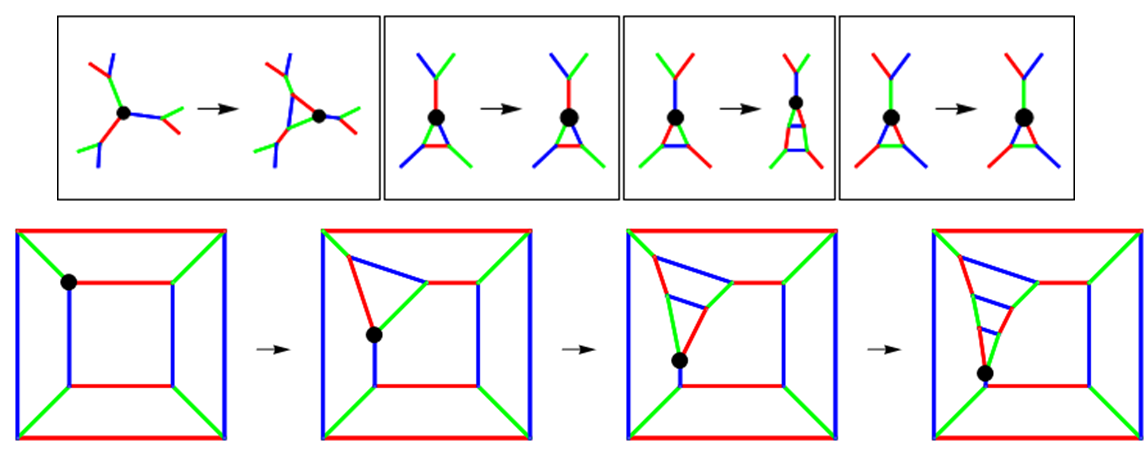}}
\vspace*{8pt}
\caption{A rule that generates simple repetitive growth evolving over the first four updates. Once the writer's surroundings are such that a pair of its neighbors are linked by a blue edge, the writer performs an update which leads to similar surroundings occurring again on the next time step.}
\label{repsimp}
\end{figure}

Let us define repetitive growth more precisely. \emph{Period $p$ repetitive growth is occurring at time $t$} when the network grows arbitrarily large eventually, and there is a distance $r$ such that the network on vertices within a distance $r$ of the writers position at time $t$ looks identical to the network on vertices within a distance $r$ of the writer at time $t+p$, and the writer never moves to an old vertex more than a distance $r-1$ from where it was at time step $t$ during the interval $[t,t+p]$. The rule shown in Figure \ref{repsimp} falls into period one repetitive growth because after the first update the structure within a distance one of the writer always looks similar. Evolving this rule for $t>0$ time steps results in a network with $8+2t$ vertices and a single triangle on the end of a long ``ladder like'' substructure.

Although one may test whether dynamics fit our definition using a computer, there are more practical ways to spot repetitive growth (see Figure \ref{rep2}). Networks undergoing repetitive growth tend to have an elongated linear, or circular shape because they are composed of a series of repeating substructures. Plotting the index of the writer over time is an excellent way to see dynamics. However it does depend on the way the vertices are indexed within the computer program, and this inevitably depends on more than just the pure topological dynamics of the system. In our case, when a vertex with index $v$ in an $L$ vertex network is replaced with a triangle, we give the vertex of this triangle with a red external edge an index $v$, and we give the other two vertices the other vertices of this triangle, with green and blue external edges, indices $L+1$ and $L+2$, respectively.


\begin{figure}
\centerline{\includegraphics[scale=0.7]{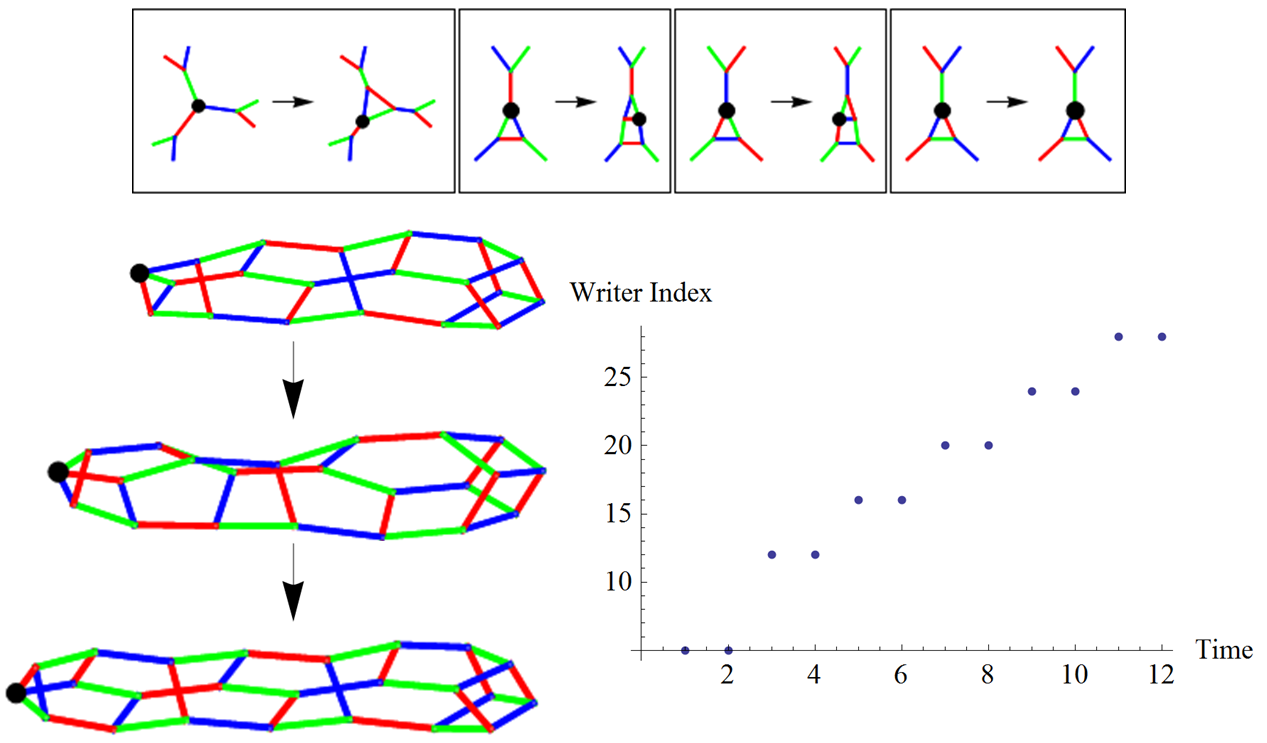}}
\vspace*{8pt}
\caption{On the left we show the networks generated on time steps $6$, $7$ and $8$ (reading downwards), when evolving from a cube. On the right we plot the index of the writer's position over time. Notice how the differences between the indices of the writer's position forms a periodic sequence. This is because, although network keeps growing, the writer keeps moving in the same relative way. This system has period two repetitive growth because the network induced on vertices within a distance two of the writer on time step $6$ looks identical to the induced on vertices within a distance two of the writer at time step $6+2$.}
\label{rep2}
\end{figure}

\subsection{Repetitive growth with long transients}

\begin{figure}
\centerline{\includegraphics[scale=0.75]{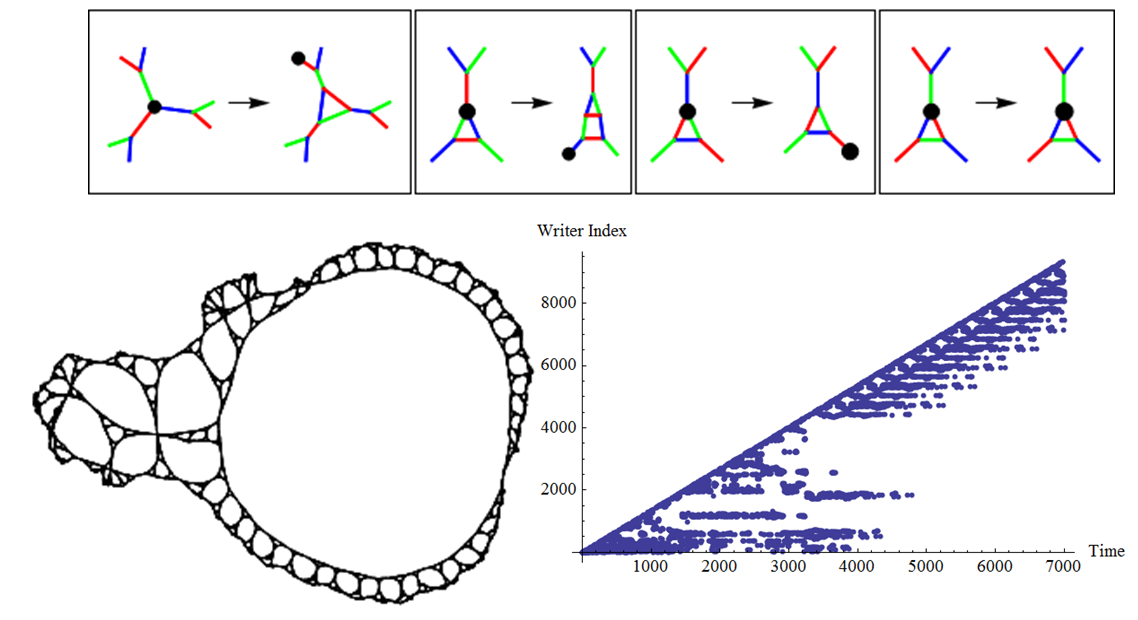}}
\vspace*{8pt}
\caption{This system has complex behaviour for the first $4994$ time steps, and then settles into period $454$ repetitive growth. We show the rule at the top.
 On the left we show an uncolored plot of the network present on time step $7000$. The large ``handle'' developing on the right of this picture is a symptom of the repetitive growth. The plot on the right shows the writer index over the first $7000$ time steps.}
\label{comp1}
\end{figure}

All but four of our repetitive growth rules settle into low (i.e., less than $12$) period repetitive growth quickly (in less than $40$ time steps). The other four rules eventually produce repetitive growth (see Figure \ref{comp1}), although it would perhaps be best to describe their behaviour as \emph{complex} because the systems have extremely long transients within which the structures appear to grow in a pseudo-random way (this is reminiscent of elementary cellular automata $110$ \cite{nks}). Interestingly, evolving the rule shown in Figure \ref{comp1} from other small networks produces different behaviour. The network known as $K_{3,3}$ is obtained by taking two clusters of three vertices, and linking each vertex in one cluster to each vertex in the other. Initiating this system from $K_{3,3}$ (instead of the cube) leads to dynamics which (again) eventually settle in period $454$ repetitive growth, although this time the transient lasts for $29964$ time steps. Another of the four rules which induce repetitive growth after a long transient is symmetrically equivalent to the rule shown in Figure \ref{comp1}, because it can be transformed into it by swapping the roles of the red and blue colors.

The other two rules which induce repetitive growth, after a long transient, are not symmetrically equivalent to the system shown in Figure \ref{comp1}, although they are symmetrically equivalent to each other. One of these rules is shown in Figure \ref{comp2}.

\begin{figure}
\centerline{\includegraphics[scale=0.75]{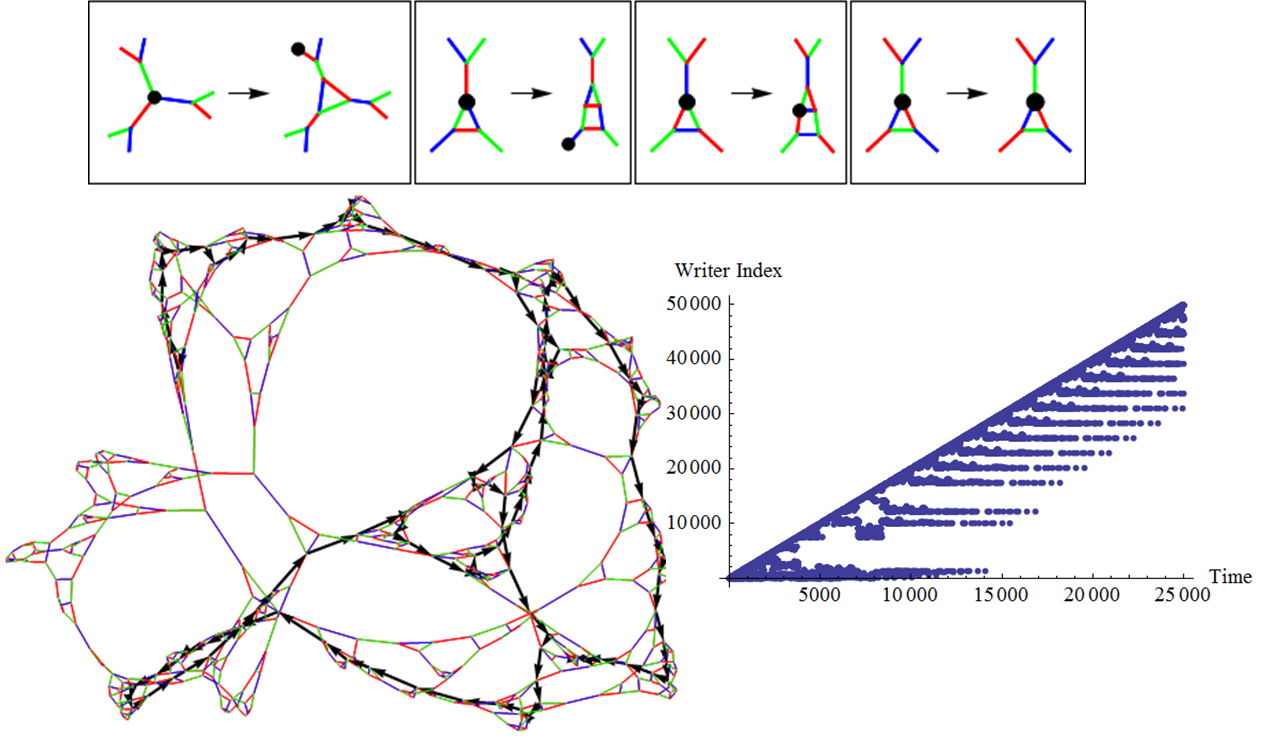}}
\vspace*{8pt}
\caption{On the left we show the network obtained after evolving this rule for $300$ time steps. The black arrows track the course that the writer has taken over the last $150$ time steps (note that the positioning of these arrows depends on how we indexed our vertices). The writer movement looks unpredictable during this early transience. On the right we plot the writer's index over time. The transient behaviour lasts $17615$ time steps, then the system settles into period $1355$ repetitive growth. Similar behaviour ensues when evolving from other initial conditions.}
\label{comp2}
\end{figure}

The ability of these systems to alter their behaviour after large amounts of time leads to the production of heterogeneous substructures.  The fact that these rules produce a complex ``ball'' of network followed by a long one dimensional structure is reminiscent of the way plants grow by complicating the structure around the initial seed and then growing out a long stalk (see the figures from section \ref{other}).

\section{Elaborate Growth}\label{class 3}

The remaining $130$ rules produce elaborate patterns of growth, leading to self similar networks. It turns out that in each of these cases, the way the writer moves is effectively one dimensional. The writer repeatedly traverses a one dimensional ``track'', applying rewrite rules as it goes, and lengthening the track with each traversal. In most cases, the fact that the writer is confined to a one dimensional substructure can be inferred directly from the rules. In the other cases this behaviour is revealed by using appropriate visualization techniques.

The writer can move backwards and forwards on the one dimensional track it is confined to. Although the writer can move in complex ways we can sort the rules into two subclasses according to the general nature of the writer's movement. Either the writer moves around and around a closed loop, in which case we say the rule has \emph{cyclic writer movement}, or the writer moves backwards and forwards over line, in which case we say the rule has \emph{bouncing writer movement}.

\subsection{A simple case with cyclic writer movement}\label{cycsimpsec}

$104$ of the rules with elaborate growth have cyclic writer movement. We show one of the simplest in Figure \ref{cycsimp}.
By looking at the rules of this system one can see that the writer must move in a one dimensional fashion. Whenever there are no red or green edges interlinking the writer's neighbors (i.e., whenever the system is not at a fixed point), the writer moves along a red edge and then a blue edge, and the writer's previous location is replaced with a triangle. The fact that the writer's movement can be expressed as a combination of steps along edges of \emph{only two} different colors implies that the writer's movement is effectively one dimensional. In particular, since any sufficiently long path along edges of alternating red/blue color must eventually return to its starting point the writer must always remain confined to a cyclic track. In this case the cyclic track starts out as the inner face of our cube (see Figure \ref{cycsimp1d}). On each update the writer replaces its current vertex with a triangle and moves two edges clockwise around the developing track of edges with alternating red/blue colors. Theorem \ref{cycsimptheorem} describes the global rewrite operation which the writer effectively performs with each complete traversal of its cyclic track. The amount of time successive traversals take doubles.

\begin{figure}
\centerline{\includegraphics[scale=0.65]{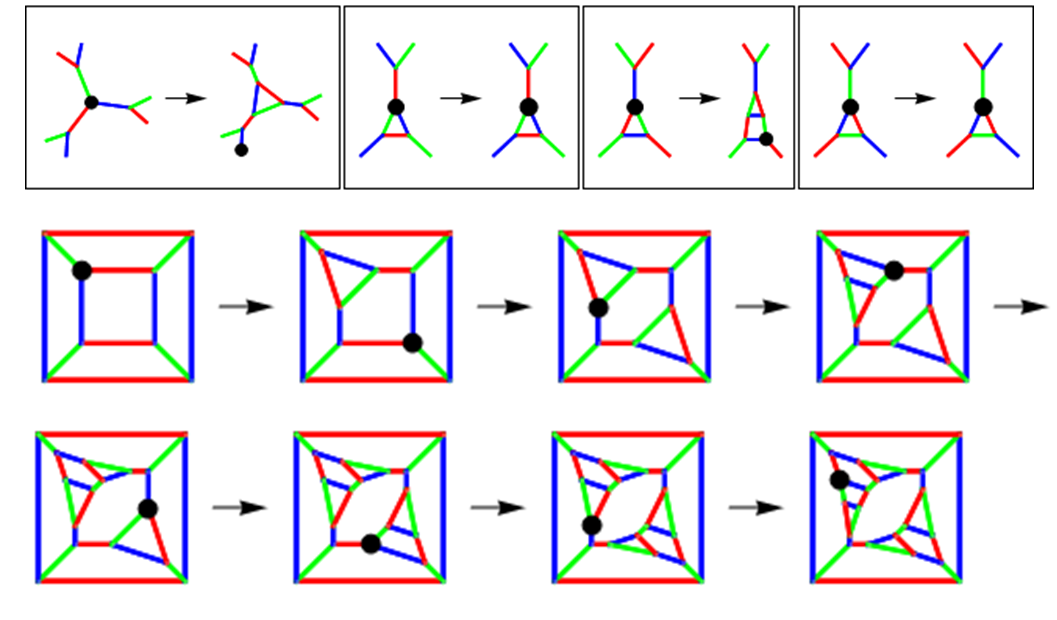}}
\vspace*{8pt}
\caption{One of the simplest rules with cyclic writer movement, evolving for seven updates.}
\label{cycsimp}
\end{figure}

\begin{figure}
\centerline{\includegraphics[scale=0.75]{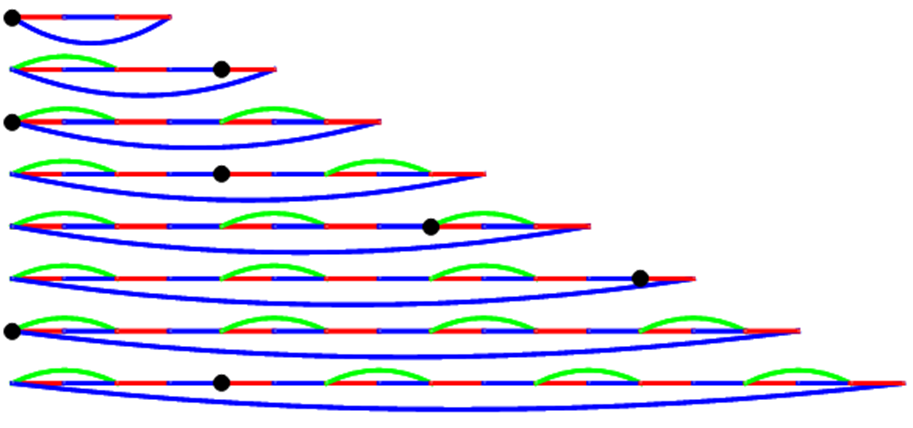}}
\vspace*{8pt}
\caption{An alternative ``one dimensional'' representation of the dynamics of the rule shown in Figure \ref{cycsimp}. Successive rows, reading downwards, show the relevant part of the network on successive time steps. Here the \emph{relevant part of the network} consists of network induced on the vertices which can be reached by moving along red or blue edges from the writer. We do not show green edges which are not part of triangles since these have no effect on dynamics. Each time step the writer replaces its current vertex with a triangle (effectively adding a red edge and a blue edge below a green arc) and moves two edges to the right.}
\label{cycsimp1d}
\end{figure}

\begin{theorem}\label{cycsimptheorem}
The way the rule shown in Figure \ref{cycsimp} evolves (with the cube as the initial condition) is such that for each $n \geq 2$, the network present on the $(2^{n+1}-2)$th time step can be obtained by taking the network present on the $(2^{n}-2)$th time step and then simultaneously replacing each vertex, with a red or blue edge interconnecting its neighbors, with a triangle (see Figure \ref{represent}).
\end{theorem}

Proofs are in the appendix.

\begin{figure}
\centerline{\includegraphics[scale=0.35]{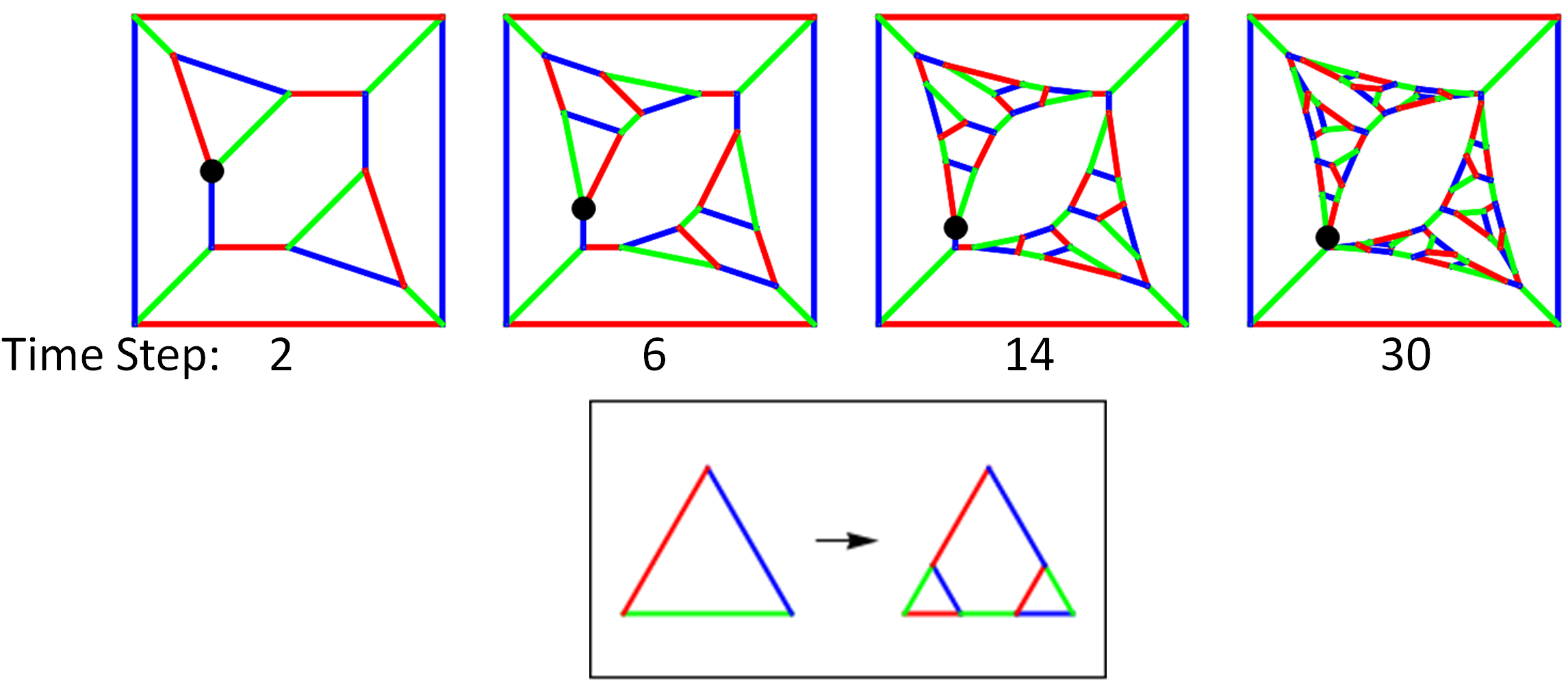}}
\vspace*{8pt}
\caption{The network generated by the system with rules shown in Figure \ref{cycsimp} on the first few time steps of the form $2^n-2: n \geq 2$. Between times of this form the writer makes a complete traversal of its cyclic path and effectively performs the rewrite operation shown at the bottom, globally.}
\label{represent}
\end{figure}

\subsection{A rule related to the golden ratio}\label{goldenrulesec}

In Figure \ref{goldenrule} we show a rule with cyclic writer movement with a complicated looking growth rate. It turns out that the growth rate can be described exactly in terms of the golden ratio $\phi = \frac{1+\sqrt{5}}{2}$.

\begin{figure}
\centerline{\includegraphics[scale=0.7]{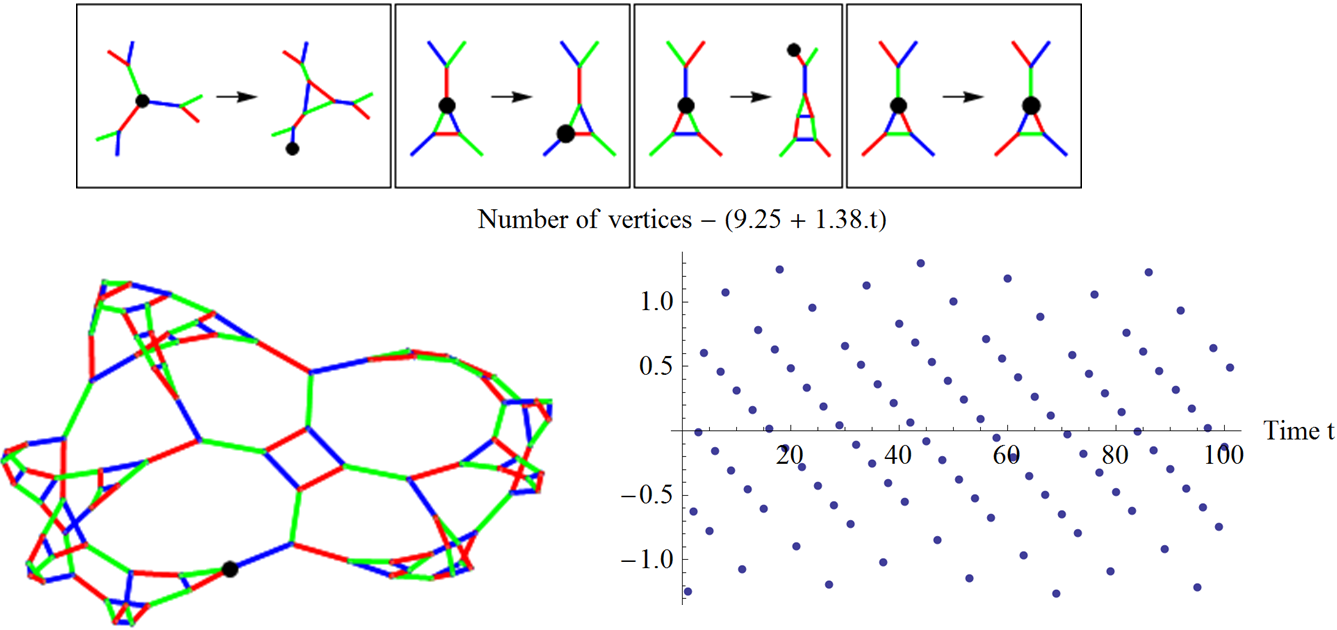}}
\vspace*{8pt}
\caption{On the left we show the network obtained by evolving this rule for $62$ time steps. The number of vertices in the network at time step $t$ grows approximately linearly with $t$. However the way the number of vertices deviates from its best linear fit looks complicated (as shown on the right).}
\label{goldenrule}
\end{figure}

\begin{theorem}\label{goldentheorem}
The number of vertices in the network obtained by evolving the rule shown in Figure \ref{goldenrule} for $t \geq 0$ time steps (starting from the cube) is $$8+2 \left\lceil \frac{1}{\phi^2} \left\lfloor \frac{t}{2} \right\rfloor \right\rceil + 2 \left\lfloor \frac{t+1}{2} \right\rfloor.$$
\end{theorem}

It is pleasing that this seemingly complex growth rate can be described simply in terms of the golden ratio, because this number appears in so many interesting places in the natural world. This colored trinet automata has similar qualitative behaviour to the one considered in subsection \ref{cycsimpsec}. Once again the writer goes around and around a cyclic track consisting of edges with alternating red/blue color. Again the system can be reduced to a one dimensional string rewrite system. The proof to Theorem \ref{goldentheorem} is based on relating this system to the binary rewrite system with rules $0 \rightarrow 01, 1 \rightarrow 011$.


\subsection{A rule with complex behaviour}\label{complexcycsec}

The rule with cyclic writer movement shown in Figure \ref{complexcyc} has a very complicated looking growth rate. Although this system can be transformed into a one dimensional rewrite system (in a similar way to the systems considered in subsections \ref{cycsimpsec} and \ref{goldenrulesec}) we have not been able to derive a formula for its growth rate.  This means this system, together with its equivalent red/blue reflection; stand out as the most complex rules in our entire set\footnote{All of the rules which generate fixed points and repetitive growth have trivial long term behaviour. Plotting the writer index over time reveals significant regularities in all the other rules exhibiting elaborate growth}. We are unable to predict the long term dynamics of these systems (although it appears the pseudo random growth pattern continues forever).


\begin{figure}
\centerline{\includegraphics[scale=0.8]{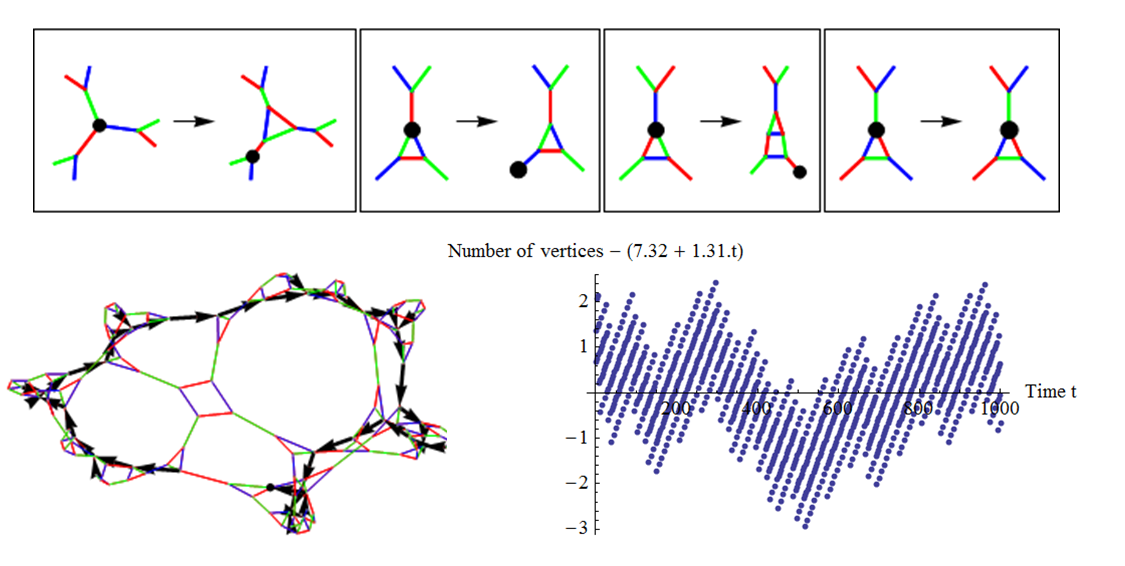}}
\vspace*{8pt}
\caption{On the left we show the structure obtained by evolving this rule for $100$ time steps. The black arrows track the positions of the writer over the previous $45$ time steps. On the right we show how the growth rate of the number of vertices deviates from its best linear fit. This plot reveals considerable complexity.}
\label{complexcyc}
\end{figure}

\subsection{A simple case with bouncing writer movement}\label{bouncesec}

The $26$ remaining rules with elaborate growth exhibit bouncing writer movement. In these rules the writer moves forwards and backwards along a one dimensional track, which grows with successive traversals. The rule in this class with the simplest behaviour is shown in Figure \ref{bounce}. This rule can be represented in a one dimensional manner similar to Figure \ref{cycsimp1d}, except that this time the writer is confined to moving on the red and green edges.


It appears that the writer effectively performs a global rewrite operation every time it makes a traversal of the one dimensional track it is confined to. Simulations suggest the that the network present on time step $2^{n+3}+n$ (where $n \geq 0$) can be obtained by taking the network present at time step $2^{n+2}+n-1$ and then simultaneously replacing each vertex with a triangle that has no red or green edges interlinking its neighbors and is not part of the external face (see Figure \ref{repbounce}).

\begin{figure}
\centerline{\includegraphics[scale=0.65]{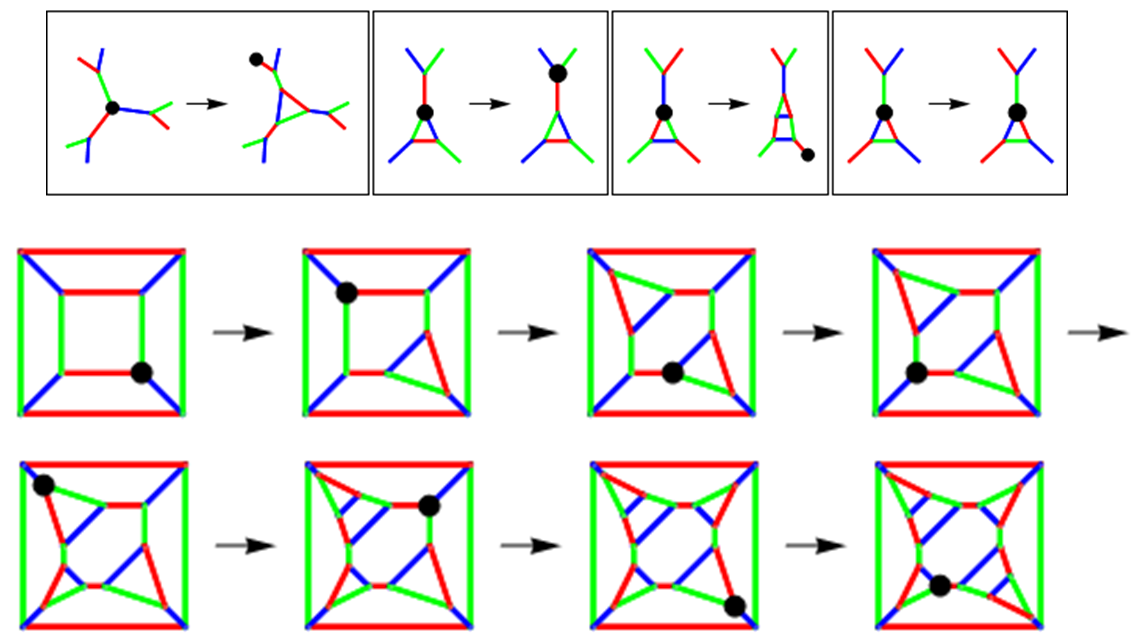}}
\vspace*{8pt}
\caption{A rule with bouncing writer movement evolving for eight updates. The writer travels around the red and green edges of the track that starts out as the central face of the cube. The writer keeps bouncing off the red edge at the bottom of this face and reversing its direction of movement.}
\label{bounce}
\end{figure}



\begin{figure}
\centerline{\includegraphics[scale=0.65]{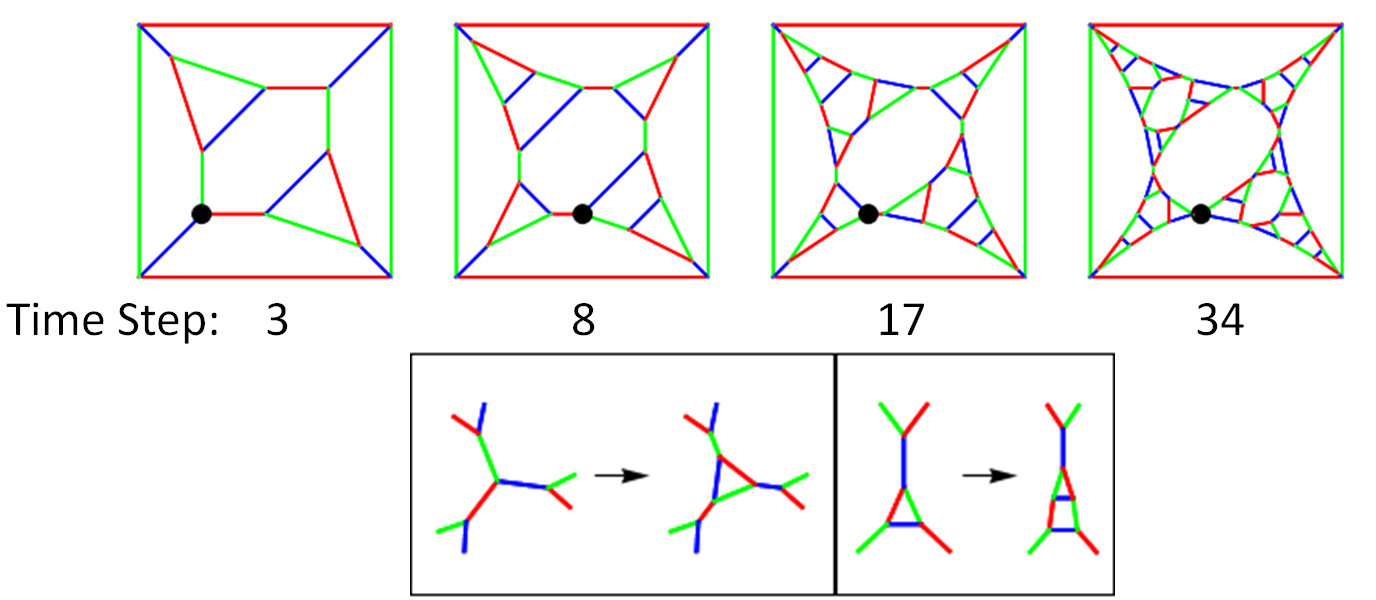}}
\vspace*{8pt}
\caption{The network generated by the system with rules shown in Figure \ref{bounce} on the first few time steps of the form $2^{n+2}-n-1: n \geq 0$. Between times of this form the writer makes a complete traversal of its linear path and both rewrite operations shown at the bottom, globally.}
\label{repbounce}
\end{figure}

Not all of the rules with elaborate growth can be reduced to one dimensional rewrite systems as directly as the examples we have considered, although the effectively one dimensional nature of the writer's movement can be revealed by plotting the trail followed by writer over several time steps (as in the left of Figure \ref{complexcyc}).

\section{More general rules}\label{other}

The behaviour of colored trinets with more general rules does not always fall into the classes listed above. The set of $3888$ rules we enumerated, and discussed above, did not allow the writer to take any action when the edge between its neighbors is green. If we remove this restriction we can find rules with four \emph{active parts} such as the one shown in Figure \ref{fouractive}. In this rule the writer continues to move in a random looking way (that is not ``one dimensional'' as in the rules with elaborate growth discussed above) for at least the first $100000$ time steps. It is an open question whether this rule eventually settles into repetitive growth.

\begin{figure}
\centerline{\includegraphics[scale=0.75]{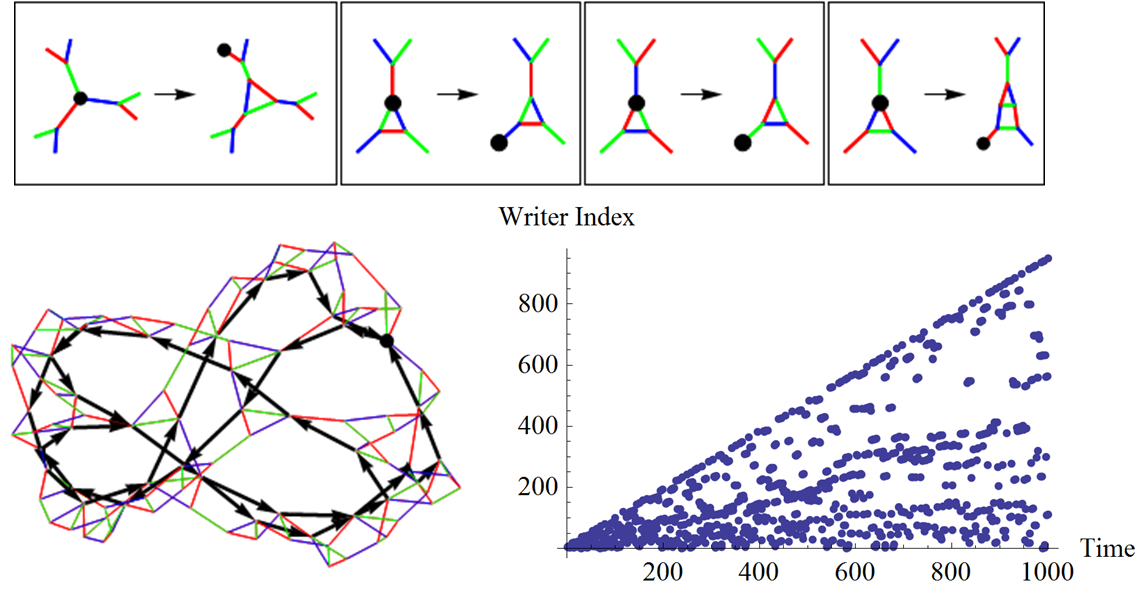}}
\vspace*{8pt}
\caption{A rule with four active parts that has complex behaviour. On the left we show the structure obtained by evolving this rule for $100$ time steps. The black arrows track the positions of the writer over the previous $45$ time steps. On the right we plot the index of the writer over the first $1000$ time steps.}
\label{fouractive}
\end{figure}

We can also consider rules which include more general kinds of rewrite operations, such as the one shown in Figure \ref{sublinear}. In this system triangles can be replaced with vertices. Unlike all of the systems in our previously considered rule set (which either has approximately linear growth rates or reached fixed points) the number of vertices in this case grows sub-linearly over time.

\begin{figure}
\centerline{\includegraphics[scale=0.45]{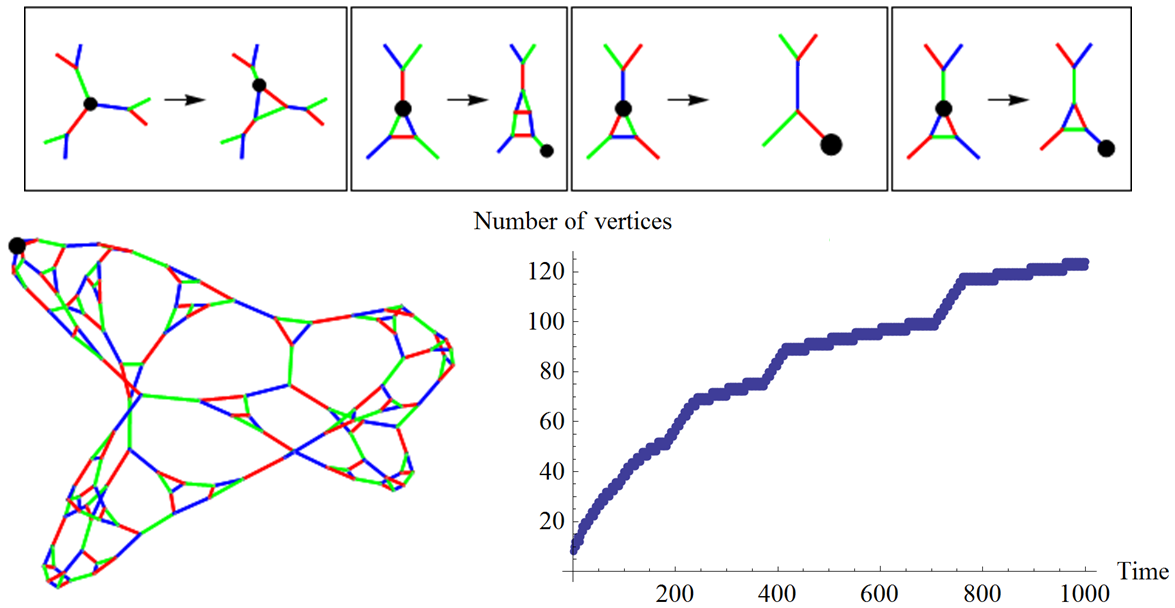}}
\vspace*{8pt}
\caption{A system with more general rules that may involve replacing a triangle with a vertex. On the left we show the network present on time step $1000$. On the right we plot the number of vertices over time.}
\label{sublinear}
\end{figure}

When exploring these more general kinds of rules, one finds many cases with sub-linear growth rates that exhibit repetitive or elaborate growth patterns. In many cases with sub-linear growth the number of vertices grows like the square root of the number of time steps elapsed.
Other rules involving triangle shrinkage can lead to networks with shapes that oscillate over time. A trivial example is shown in Figure \ref{shrink} (although rules exist with much higher period oscillations).

\begin{figure}
\centerline{\includegraphics[scale=0.35]{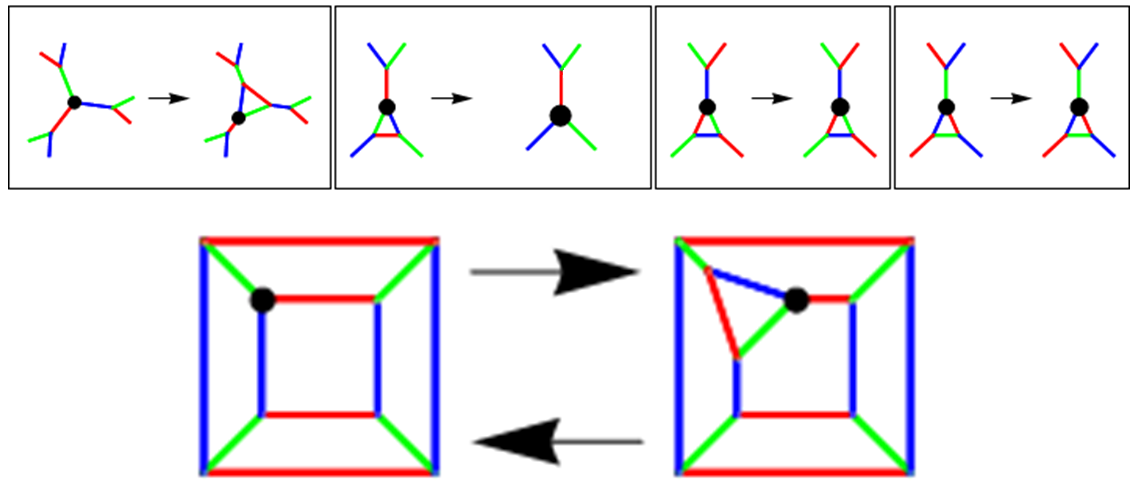}}
\vspace*{8pt}
\caption{In this system the network shape changes periodically.}
\label{shrink}
\end{figure}

\begin{figure}
\centerline{\includegraphics[scale=0.45]{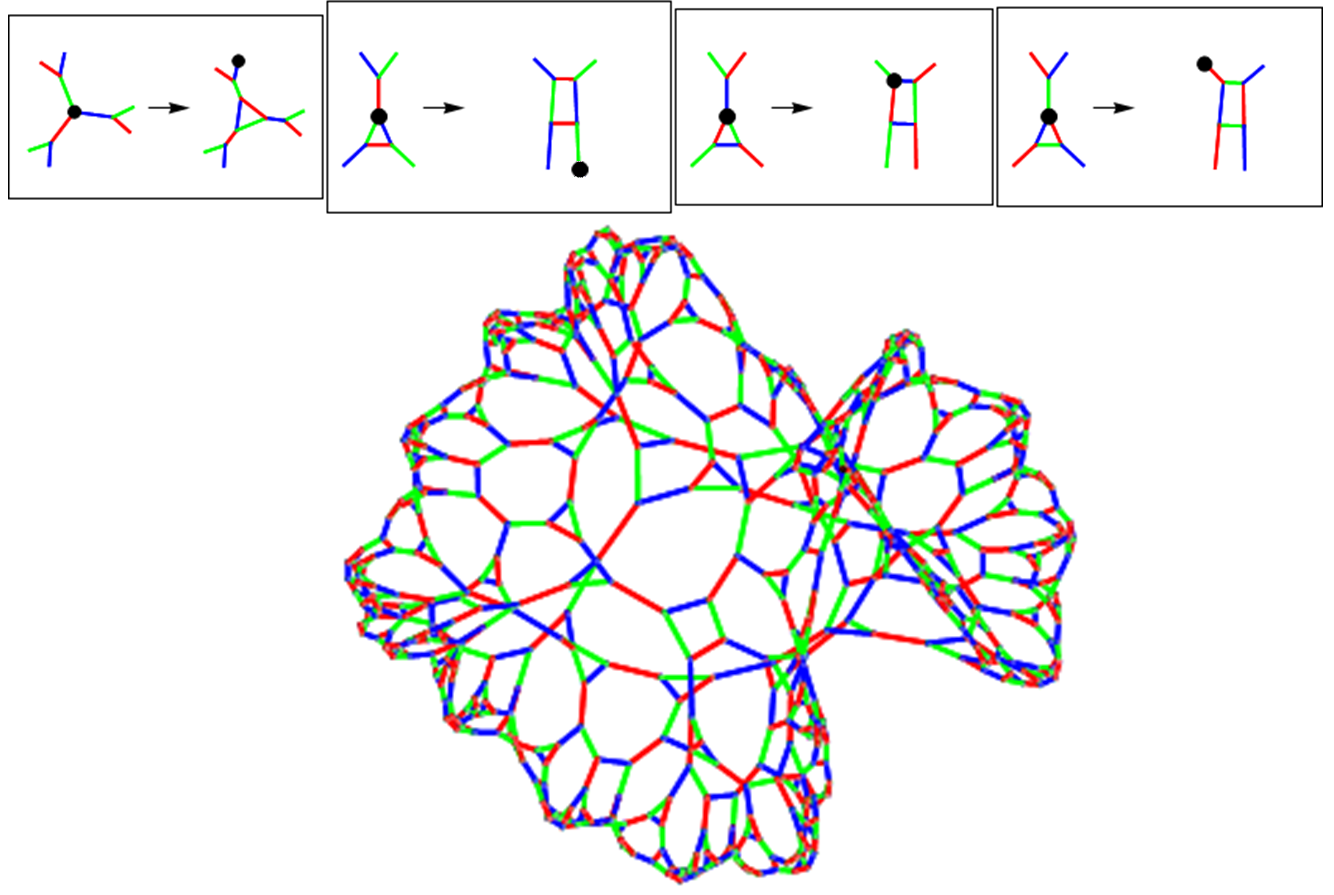}}
\vspace*{8pt}
\caption{A system which includes the \emph{exchange} rule, which effectively replaces a triangle with a square. We show the rule at the top. The structure is generated by running this system for $500$ time steps.}
\label{hyp}
\end{figure}

We can also consider rules including the so called \emph{exchange} operation, which rewires an edge \cite{jap}. In Figure \ref{hyp} we show a rule which uses this type of exchange operation. The rule depicted in Figure \ref{hyp} has some interesting spatial properties. For a vertex $v$, the number of vertices a distance $k$ from $v$ grows exponentially with $k$. This suggests the rule generates hyperbolic (i.e., negatively curved) space. Indeed, this can be verified using the techniques described in \cite{grom} and \cite{curve}. According to these works, a connected network $G$ with vertex set $V(G)$ is \emph{scaled Gromov hyperbolic} (which means it is negatively curved on a large scale) when $\frac{\delta(G)}{\dia (G)} < \frac{3}{2}$. Here
$$\delta (G) = \max\{\min\{d(x,y)+d(y,z)+d(z,x): x \in [b,c], y\in [a,c], z \in [a,b] \}, a,b,c \in V(G) \}.$$
Where $[i,j]$ denotes the shortest path from $i$ to $j$, and $d(i,j)$ denotes the distance from $i$ to $j$. Also $\dia (G) = \max\{d(u,v): u, v \in V(G)\}$ is the diameter of $G$. In our case, the network $G$ obtained by evolving the rule shown in Figure \ref{hyp} for $100$ time steps has $\frac{\delta(G)}{\dia (G)} =1.28571$, and so it is indeed scaled Gromov hyperbolic.

The $3888$ rules in the set that we initially focused on do not generate negatively curved space. This can be seen by noting the each of these rules generates planar networks, and planar networks have non-negative combinatorial curvature\footnote{Combinatorial curvature is a local measure, but the fact none of these networks have negative curvature on the large scale was verified by checking that each network which grows arbitrarily large has a well defined node shell dimension \cite{cause} less than $3$.} (see the Gauss-Bonnet formula \cite{bonnet}). These rules preserve planarity because the only structural modification they allow is the replacement of a vertex with a triangle (and this can be thought of as chopping the corner off a polyhedron -a planarity preserving operation). On the other hand, the exchange operations within the rule depicted in Figure \ref{hyp}, allow it to generate a non-planar network, when initiated from the cube.





\begin{figure}
\centerline{\includegraphics[scale=0.35]{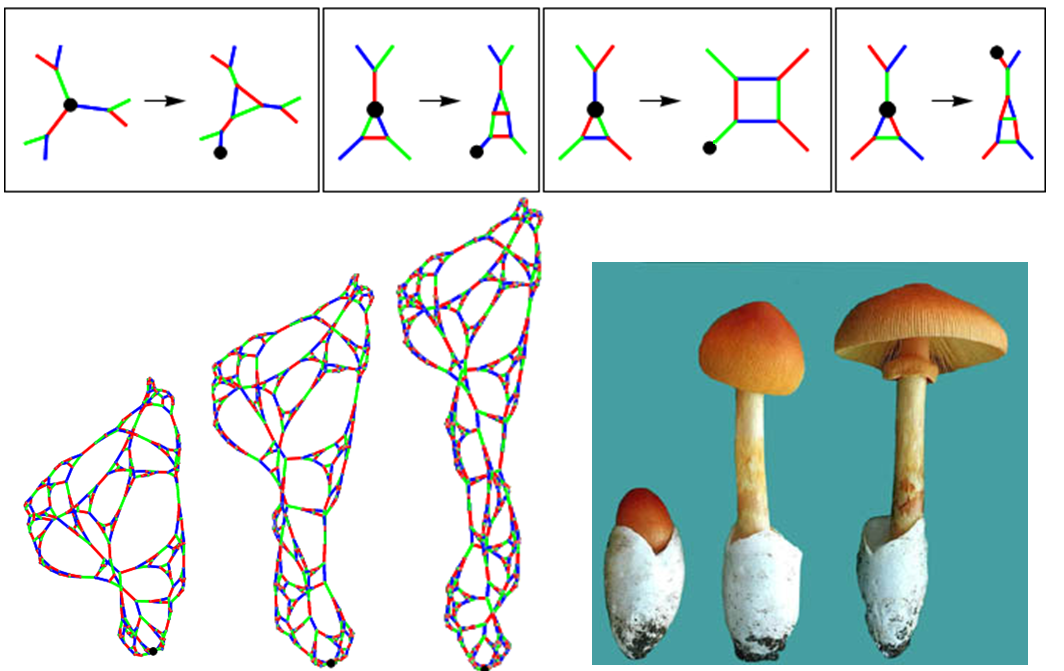}}
\vspace*{8pt}
\caption{On the left we show the networks generated by evolving this system for $150$, $200$ and $250$ time steps, evolving from the cube. On the right we show a photo depicting the development stages of a mushroom, taken from \cite{mush}. The cap is the first part of the mushroom to develop, and then a long stem grows our from the spore at the bottom. Our system generates a ball like network initially, and then grows out a linear structure. These pictures could be made to look more alike by manipulating the network embedding. The networks are drawn using Mathematica's default \emph{Spring Electrical Embedding}}
\label{mushroom}
\end{figure}

\subsection{The resemblance with real world objects}

It is remarkable how varied the forms produced by trinet automata with very similar rules can be. What is even more remarkable is that pictures similar to a diverse range of life forms and objects can easily be produced. In Figures \ref{turtle}, \ref{caterpiller}, \ref{seashell},\ref{brain} and \ref{necklace} we show networks (produced by evolving colored trinet automata, from the cube) next to real world objects.  In many cases there is significant similarity.
When looking at these pictures it is important to keep in mind that a trinet may be represented by a wide variety of pictures, because the only relevant information is the connectivity. In other words, one could draw the colored trinets in a different way, and their resemblance to these objects would disappear. Never the less, each of these trinet pictures where produced using Mathematica's standard, \emph{spring embedding} or \emph{spring electrical embedding} graph drawing algorithm. This means the way the networks are drawn was chosen by the plotting algorithm to reflect the topology of the network. In this regard, one can argue that there is similarity between the real objects shown and the colored trinets pictured because the most natural ways to draw the networks yield pictures that look like the objects. Figures \ref{turtle}, \ref{caterpiller}, \ref{seashell},\ref{brain}, and \ref{necklace} also show topological similarities between the colored trinets and the real objects.

\begin{figure}
\centerline{\includegraphics[scale=0.35]{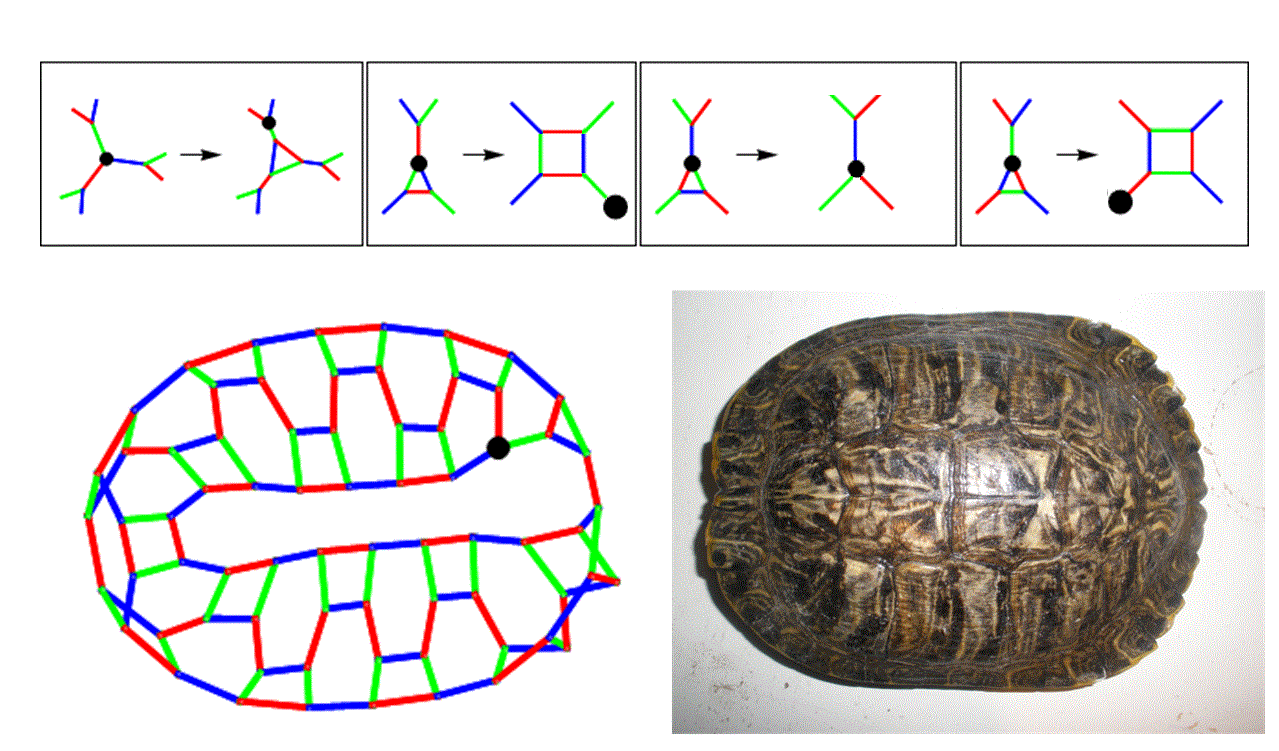}}
\vspace*{8pt}
\caption{On the left we show the network obtained from evolving our rule for $46$ time steps, starting from the cube. The network drawn using mathematica's standard \emph{spring embedding} algorithm. On the right we show a photo of a turtle shell taken from \cite{turtle}.}
\label{turtle}
\end{figure}

\begin{figure}
\centerline{\includegraphics[scale=0.35]{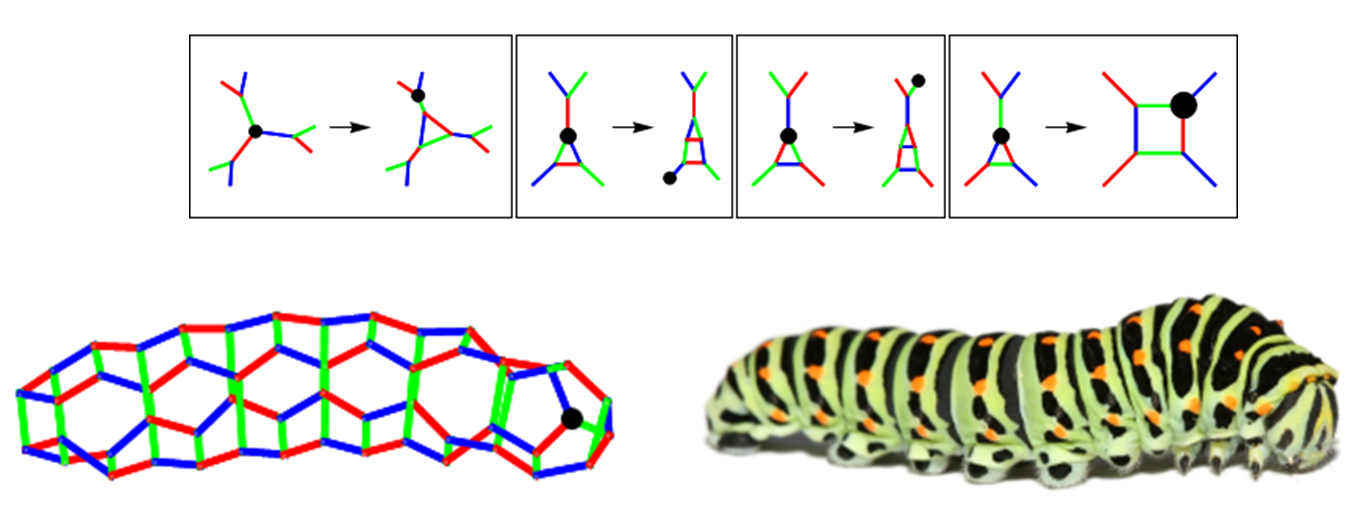}}
\vspace*{8pt}
\caption{On the left we show the network obtained from evolving our rule for $34$ time steps, starting from the cube. The network drawn using mathematica's standard \emph{spring embedding} algorithm. On the right we show a photo of a caterpillar taken from \cite{caterpiller}.}
\label{caterpiller}
\end{figure}

\begin{figure}
\centerline{\includegraphics[scale=0.35]{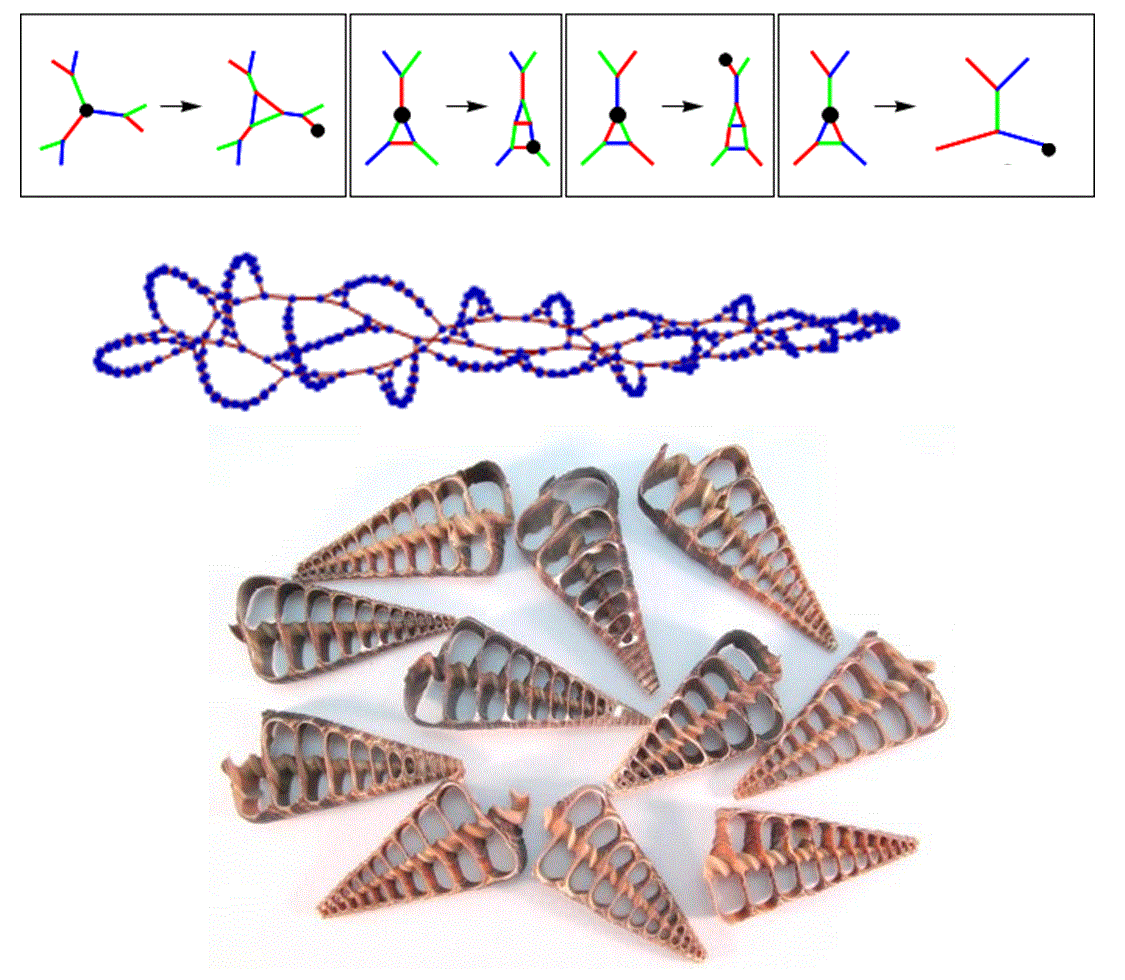}}
\vspace*{8pt}
\caption{At the top we show the rule. In the middle we show the network obtained from evolving our rule for $1600$ time steps, starting from the cube. The network drawn using mathematica's standard \emph{spring electrical embedding} algorithm (we have not shown edge colors). At the bottom we show a photo of sliced Telescopium seashells taken from \cite{seashell}.}
\label{seashell}
\end{figure}

\begin{figure}
\centerline{\includegraphics[scale=0.35]{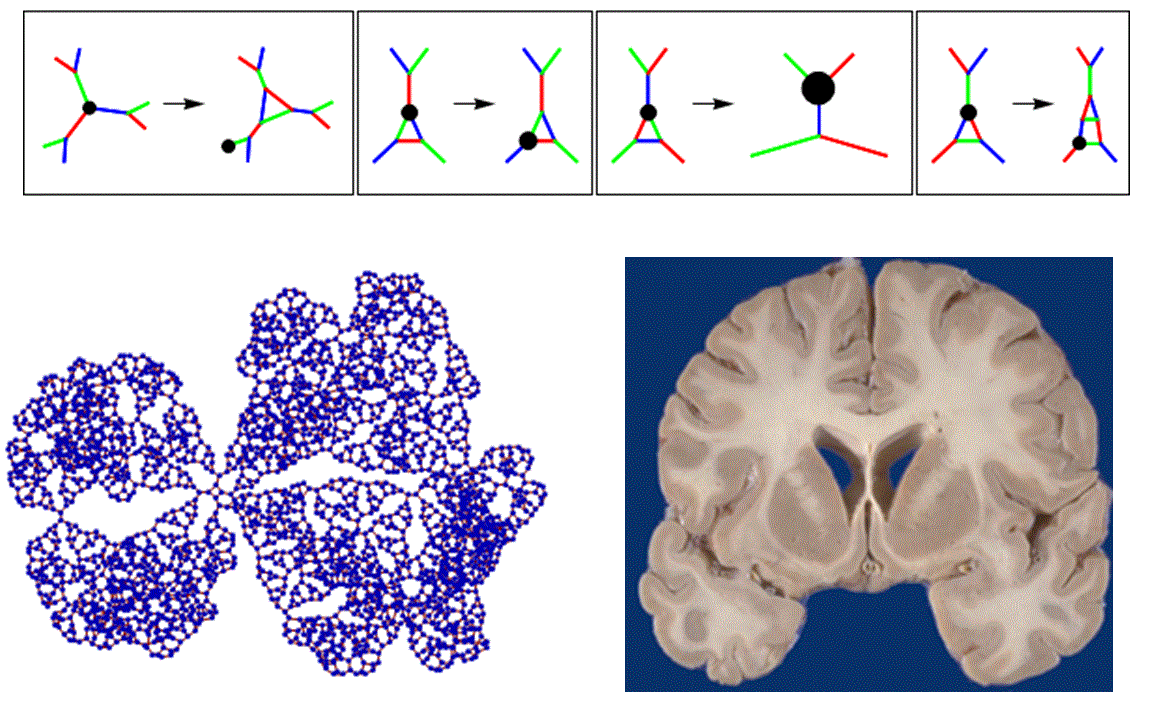}}
\vspace*{8pt}
\caption{On the left we show the network obtained from evolving our rule for $1430$ time steps, starting from the cube. The network drawn using mathematica's standard \emph{spring embedding} algorithm (we have not shown edge colors). On the right we show a photo of a coronal slice through a human brain taken from \cite{brain}.}
\label{brain}
\end{figure}

\begin{figure}
\centerline{\includegraphics[scale=0.35]{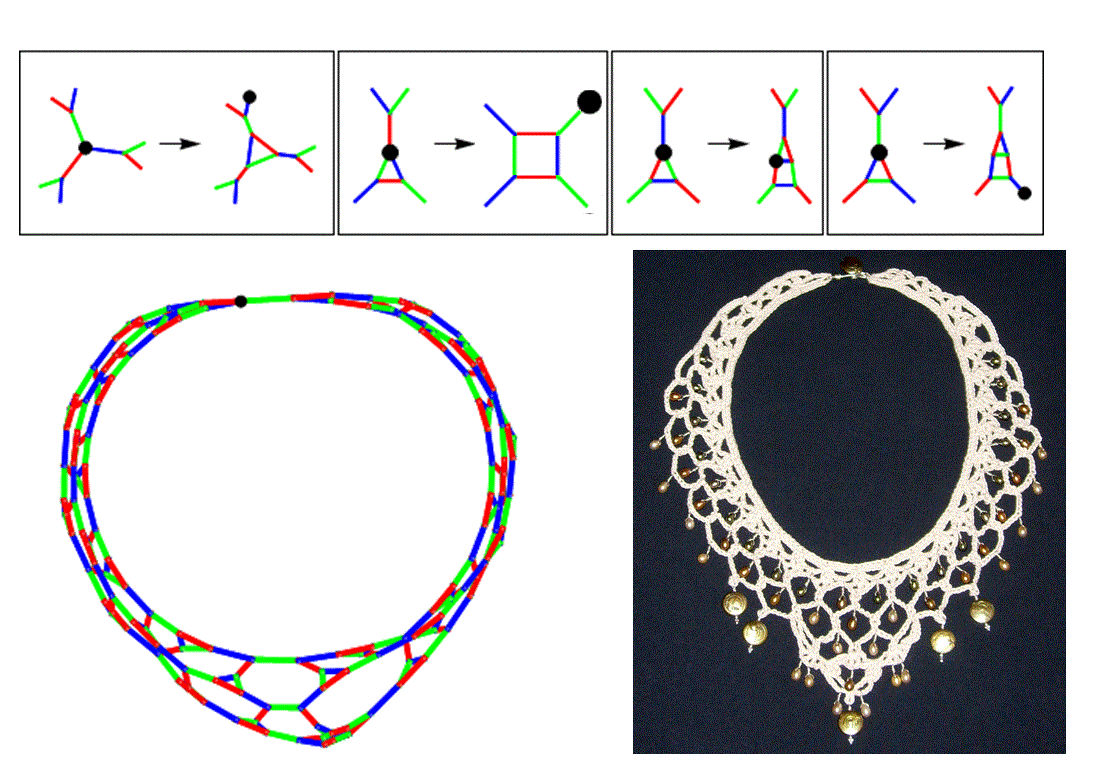}}
\vspace*{8pt}
\caption{On the left we show the network obtained from evolving our rule for $100$ time steps, starting from the cube. The network drawn using mathematica's standard \emph{spring electrical embedding} algorithm. On the right we show a photo of a Necklace made from crochet lace, pearls, and sterling silver taken from \cite{necklace}.}
\label{necklace}
\end{figure}

\section{Conclusion}

We have explored many aspects of the behavior of colored trinet automata. We found that the systems in our initial rule space could be sorted into three classes - fixed, repetitive growth and elaborate growth. We have shown that each class can be described in terms of writer movement.
When we look at the dynamics of more general rules we find other types of behavior. In particular, we find rules which generate sub-linear growth, persistent complex behaviour, periodically changing networks and hyperbolic space. We have also noted that many trinets with more general rules can produce forms remarkably similar to a diverse range of real world objects.

As minimal models capable of growing complex structure, these systems should have some applications. The systems could be realized by creating a robot which pulls itself along chains, ropes or silk and drops lines (like a spider) to triangulate vertices. This type of realization could perhaps be useful for weaving or construction work. The ability of these systems to produce exotic network structures with relative ease could make them useful in network design. In existential graph theory there are many open questions as to whether large trinets exist with particular properties and it will be interesting to see whether our systems can resolve any of these issues.

It will also be interesting to see if these systems have any applications in modeling. The way structures grow under repetitive growth is reminiscent of the way plants grow (see Figure \ref{mushroom}). Many of the networks produced by rules with elaborate growth rules seem reminiscent of polygonal networks in cracked earth or foams. The general idea of having a single writer that rewires a complex network is reminiscent of brains and databases where the connectivity is altered to store information.

There many directions that this work can be taken in the future. We have explored the dynamics of many colored trinet automata and identified many kinds of behavior, however a more general classification scheme is evidently required to deal with the systems we described in section \ref{other}. Exploring more general rules will be exciting, and will surely yield systems with other interesting kinds of behavior and applications. There are also many unanswered questions about the systems presented here, for example: is there a simple formula for the growth rate of system shown in Figure \ref{complexcyc}? Does the complex behavior in Figure \ref{fouractive} persist forever ? How can rules be characterized according to the structural properties of the networks they produce ?


\section*{Acknowledgments}
This work is supported by the General Research Funds (Project Number 412509) established under the University Grant Committee of the Hong Kong Special Administrative Region, China.

\appendix
\section{Appendix}
A word $w$ with alphabet $\mathcal{A}$ is a sequence $w_0...w_{m-1}$ of $m \geq 0$ elements $w_i \in \mathcal{A}$ from $w$. The length $|w|$ of such a word is the number elements in it ($m$ in this case). The empty word is a word of length $0$. Notice that we index the elements in our words $w$ from $0$ to $|w|-1$.
Let $\mathcal{A}^*$ denote the set of all words with alphabet $\mathcal{A}$. Let
\begin{equation}wx=w_0...w_{|w|-1}x_0...x_{|x|-1}\end{equation}
denote the word formed by concatenating words $w$ and $x$. If $k >0$ then $w^k$ denotes the word obtained by concatenating $k$ copies of $w$ together (e.g. $w^3 = www$). Also $w^0$ is the empty word.

We let $r, b $ and $g$ denote the colors red, blue and green, respectively. We define the \emph{type} $\rho(v)_G$ of a vertex $v$ in a colored trinet $G$ such that, if none of $v$'s neighbours are interlinked then $\rho(v)_G =0$ and if exactly one pair of $v$'s neighbors are interlinked, by an edge of color $c \in \{r,b,g \}$ then $\rho (v)_G =c$.

We let $(X^t,G^t)$ denote the state of a system at time step $t$ (i.e., after $t$ updates, staring from the cube). Here $X^t$ is the writer vertex on time step $t$ and $G^t$ is the colored trinet on time step $t$. In each of the systems we discuss, we start with a cube as our initial network, $G^0$ and each structural modification made to the network during an update involves replacing a single vertex with a triangle. It follows that no vertex will have more than one pair of interlinked neighbors, in any network that occurs during the evolution of one of our systems. It follows that the type, $\rho(v)$ of each vertex is well defined. In our systems the current type of the writer vertex determines how the writer will move and whether or not it will replace its old location with a triangle.


The following notation will help us to describe motion. For a word $w \in \{r,b,g\}^*$ we let $v[w]$ denote the vertex one arrives at by starting at $v$ and following the sequence of edges whose colors are described by $w$.
More precisely, for a color $c \in \{r,b,g \}$ and vertex $v$ in a colored trinet $G$ let $v[c]$ denote the vertex linked to $v$ by an edge of color $c$. If $w$ is the empty word, we let $v[w]=v$. For a word of the form $xc$ (where $x \in \{r,b,g\}^*$ and $c \in \{r,b,g \}$) we let $v[xc] = v[x][c]$.
\subsection{Proof of Theorem \ref{cycsimptheorem}}
Let $(X^t, G^t)$ denote the state of the system on time step $t$, evolving under the rules shown in Figure \ref{cycsimp} (starting with $G^0$ as the cube).
Let $U(t)$ denote the sequence $U(t)_0...U(t)_{|U(t)|-1}$ of vertices in $G^t$ that one reaches by starting at the writer vertex $X^t$ and moving along edges with color alternating between red and blue. More precisely, $\forall k \in \{0,1,...,|U(t)|-1\}$ we have
 \begin{eqnarray}
U(t)_k = \left\{\begin{array}{ll}
X^t[(rb)^{k/2}] & \text{if } $k$ \text{ is even},\\
X^t[r(br)^{(k-1)/2}] & \text{if } $k$ \text{ is odd.}
\end{array}\right. ~
\end{eqnarray}
The length of this sequence, $|U(t)| =2. \min\{k>0:X^t[(rb)^k]=X^t \}$, is the number of vertices in the cycle of edges with colors alternating between red and blue, that contains the writer. Let
\begin{equation}J(t) = J(t)_0...J(t)_{|J(t)|-1} = \rho(U(t)_0)_{G^t}...\rho(U(t)_{|U(t)|-1})_{G^t}\end{equation}
be the sequence of types of vertices in $U(t)$. In other words, $|J(t)|=|U(t)|$ and $J(t)_k = \rho(U(t)_k)_{G^t}$, $\forall k \in \{0,1,...,|J(t)|-1 \}$.
\begin{lemma}\label{1de}

$\forall T \geq 1$, we have $J(2T)=(bgr0)^{T+1}$ and $J(2T+1)=00(bgr0)^{T+1}$.
\end{lemma}

\begin{proof}
For $T=1$ we have $J(2) = bgr0bgr0$ and $J(3) = 00bgr0bgr0$, and so the result holds in this case.

Suppose the result holds for some $T \geq 0$.

Now, $J(2T+1)=00(bgr0)^{T+1}$, by hypothesis. This means the writer has type $\rho(X^{2T+1})_{G^{2T+1}}=J(2T+1)_0 = 0$ on time step $2T+1$. Now, by consulting the rules behind the system one can see that updating, from this state, will consist of the writer moving along a red edge, and then a blue edge (to the position of vertex $U(2T+1)_2$) and then replacing its previous location with a triangle. This implies that the type sequence $J(2T+2)$, for the next time step, can be constructed by taking $J(2T+1)$ and rotating it two steps to the left (to represent to writer's movement) to get $(bgr0)^{T+1}00$, and then replacing the writers previous position (i.e., vertex corresponding to the penultimate element in $(bgr0)^{T+1}00$) with a triangle ($bgr$) to get $J(2T+2) = (bgr0)^{T+1}bgr0 = (bgr0)^{T+2}$.

On time step $2T+2 = 2(T+1)$ the writer has type $\rho(X^{2(T+1)})_{G^{2(T+1)}} = J(2(T+1))_0 = b$. Now, by consulting the rules behind the system one can see that updating, from this state, will consist of the writer moving along a red edge, and then a blue edge (to the position of vertex $U(2(T+1))_2$) and then replacing its previous location with a triangle. This implies that the type sequence $J(2(T+1)+1)$, for the next time step, can be constructed by taking $J(2(T+1))$ and rotating it two steps to the left (to represent to writer's movement) to get $r0(bgr0)^{T+1}bg$, and then replacing the writers previous position (i.e., vertex corresponding to the penultimate element in $r0(bgr0)^{T+1}bg$) with a triangle ($bgr$). It follows that the type sequence on the next time step will be $J(2(T+1)+1) = 00(bgr0)^{T+1}bgr0 = 00(bgr0)^{T+2}$. The reason for this is that, immediately before the update, the writer's previous position (which corresponds to the penultimate element in the word $r0(bgr0)^{T+1}bg$) is part of a pre-existent triangle (together with the vertices corresponding to the last and first elements in the word $r0(bgr0)^{T+1}bg$), and when the writers previous position is replaced with a triangle this pre-existent triangle becomes a square (thus effecting the types of the other vertices which used to be part of it). The effect of this is that we replace the rightmost $b$ in our word $r0(bgr0)^{T+1}bg$ with $bgr$ (to represent the replacement of the writers previous position with a triangle) and we replace the rightmost $g$ and the leftmost $r$, in our word, with $0$ to represent how the types of these vertices change when the triangle they used to be part of is expanded apart by the addition of the new triangle. These modifications lead to the word $J(2(T+1)+1) = 00(bgr0)^{T+1}bgr0$.

Now we have shown that our result holds for $T=1$, and if our result holds for any $T \geq 1$ then it holds for $T+1$. We can hence prove the result holds $\forall T \geq 1$ using induction with $T$.

\end{proof}

\begin{corollary}\label{coro1}
$\forall t \geq 0$, the update at time step $t$ consists of the writer moving along a red edge and then a blue edge, and replacing its previous position with a triangle.
\end{corollary}
\begin{proof}
For $t=0$ and $t=1$ this result is clearly true. Now suppose $t \geq 2$.

If $t$ is even then $\exists T \geq 1$ such that $t = 2T$, and so Lemma \ref{1de} implies that $J(t) = J(2T) = (bgr0)^{T+1}$, and so $\rho(X^t)_{G^t} = J(2T)_0 = b$, and so (according to the rules of our system) the update at time step $t$ consists of the writer moving along a red edge and then a blue edge, and replacing its previous position with a triangle.

If $t$ is odd then $\exists T \geq 1$ such that $t = 2T+1$, and so Lemma \ref{1de} implies that $J(t) = J(2T) = 00(bgr0)^{T+1}$, and so $\rho(X^t)_{G^t} = J(2T)_0 = 0$, and so (according to the rules of our system) the update at time step $t$ consists of the writer moving along a red edge and then a blue edge, and replacing its previous position with a triangle.
\end{proof}

We say that $n \geq 2$ satisfies property (1) when $J(2^n -2) = (bgr0)^{2^{n-1}}$, and for each vertex $v$ in $G^{2^n -2}$ of type $\rho(v)_{G^{2^n -2}} \in \{r,b\}$ there is a constant $q(v) \in \{0,1,...,2^n-1\}$ such that $v = X^{2^n -2}[(rb)^{q(v)}] = U(2^n-2)_{2q(v)}$ (in other words, each vertex of type $r$ or $b$ can be reached by repeatedly doing \emph{$rb$ hops} from the writer's current position $X^{2^n -2}$). Here an $rb$ hop consists of moving along a red edge, and then a blue edge.

Suppose that $n \geq 2$ satisfies property (1). We will show that this implies that $n+1$ satisfies property (1), and that $G^{2^{n+1}-2}$ can be obtained by taking $G^{2^{n}-2}$ and simultaneously replacing each vertex of type $r$ or $b$ with a triangle.

Now according to Corollary \ref{coro1} the writer will visit the vertices
$$X^{2^n -2}[(rb)^0], X^{2^n -2}[(rb)^1],...,X^{2^n -2}[(rb)^{2^n-1}]$$
in sequence, over the next $2^n$ updates, and each such vertex will be replaced with a triangle. It follows that the network $G^{2^n-2+2^n}=G^{2^{n+1}-2}$ present $2^n$ time steps later can be obtained by taking $G^{2^n-2}$ and simultaneously replacing each vertex of type $r$ or $b$ with a triangle. Moreover, each vertex $v'$ of $G^{2^{n+1}-2}$ which is of type $r$ or $b$ will be part of a triangle which was created during an update which occurred upon a time step in $\{(2^n-2),(2^n-2)+1,...,(2^n-2)+2^n-1\}$. It follows that each such vertex $v'$ will be reachable by doing some number $q'(v') \geq 0$ of $rb$ hops from the writer's current position $X^{2^{n+1}-2}$. Moreover, Lemma \ref{1de} implies that $J(2^{n+1}-2) = (bgr0)^{2^{n}}$. This implies that each of our vertices $v'$ (with type $r$ or $b$) can be reached by doing a number of $rb$ hops $q'(v')$ such that $q'(v) \leq 2^{n+1}-1$. It follows that $n+1$ satisfies property (1).

Now $n=2$ clearly satisfies property (1). It follows, by induction with $n$ that $\forall n \geq 2$ we have $n$ satisfies property (1) and $G^{2^{n+1}-2}$ can be obtained by taking $G^{2^{n}-2}$ and simultaneously replacing each vertex of type $r$ or $b$ with a triangle. QED.

\subsection{A complete description of the behaviour of the rule shown in Figure \ref{cycsimp}}

Theorem \ref{cycsimptheorem} shows how the system pictured in Figure \ref{cycsimp} can be viewed as evolving via global replacement operations. Under this system we can describe exactly how the network structure and writer position change with time.

Let $\mathbb{N}$ and $\mathbb{N}_0$ denote the sets of positive and non-negative integers respectively.

For any pair of non-negative integers $x,y \in \mathbb{N}_0$ let $\zeta(x,y) = \min\{ x+2y,2x\}$.

We describe the state of a trivalent network automata by listing the vertex set, set of red edges, set of blue edges, set of green edges and the position of the writer, in that order.

Suppose $T \in \mathbb{N}: T \geq 2$. Let $n(T) = \max\{m \in \mathbb{N}: 2^m -2 \leq T \}$ and $t(T) = T-(2^{n(T)} -2)$.

Now let us define the state $H(T) = (V(T), E_r(T), E_b(T), E_g(T), W(T))$ at time $T$ such that:

\begin{enumerate}
\item The set of vertices is
$V(T) = \{-4,-3,..,2^{n(T)+1}+2.t(T) - 1\}.$

\item The set of red edges is
$E_r(T) = \{ \{-1,-4\}, \{ -3, -2\} \} \cup \{ \{i, i+1 \mod 2^{n(T) + 1} + 2.t(T) \} : i \in \{ 1,3,..,2^{n(T) + 1} + 2.t(T) -1\}\}.$

\item The set of blue edges is
$E_b(T) = \{ \{-4,-3\}, \{ -2, -1\} \} \cup \{ \{i, i+1\} : i \in \{ 0,2,..,2^{n(T) + 1} + 2.t(T) -2\}\}.$

\item The set of green edges is
$E_g(T) = \{ \{-4,0\}, \{ -3, \zeta(2^{n(T)-1},t)\}, \{ -2, \zeta(2^{n(T)},t)\} , \{ -1, \zeta(2^{n(T)+1} - 2^{n(T)-1},t)\} \} \cup
\{ \{ \zeta( 2^{\alpha +2}. \beta +2^\alpha , t), \zeta( 2^{\alpha +2}. \beta +2^\alpha +2^{\alpha +1} , t) \} : \alpha \in \{0,1,..,n(T)-2\}, \beta \in \{0,1,..,2^{n(T)- \alpha -1} -1\} \}
\cup
\{ \{4m+1,4m+3\} : m \in \{0,1,..,t(T)-1\}\}.$

\item The writer position is $W(T) = 4.t(T)+1$.

\end{enumerate}

Now we claim, for each $T \geq 2$, that $H(T)$ is equivalent to the state of $S(T)$ of the system obtained by evolving the rule shown in Figure \ref{cycsimp} for $T$ time steps, starting from the cube. More precisely, we claim that there is a graph isomorphism from the network of $H(T)$ to the network of $S(T)$ that preserves the position of the writer and the colors of the edges. Such an isomorphism is a bijection from $V(T)$ to the vertex set of the network in state $S(T)$ which sends the writer position $W(T)$ to the position of the writer in $S(T)$, and which has the property that, a pair of vertices are linked by an edge of a particular color, in $H(T)$, if and only if the images of these vertices under our bijection, are linked by an edge of the same color.

This equivalence of $H(T)$ and $S(T)$ can be seen by induction. It can easily be verified in the case where $T=2$. Suppose the result holds for some generic $T \geq 2$ (i.e., suppose that $H(T)$ is equivalent to $S(T)$). Now we simply have to show that the result also holds for $T+1$, and we can use induction with $T$ to complete the proof. We can show the result holds for $T+1$ by showing that the state $H^*$ obtained by updating $H(T)$ under the system shown in Figure \ref{cycsimp} is isomorphic to $H(T+1)$.

The state $H^*$ obtained by updating $H(T)$ is equal to the state obtained by taking $H(T)$, and moving the writer along a red edge followed by a blue edge, and then performing a structural modification upon the network. This structural modification (which yields the network of $H^*$ from the network of $H(T)$) consists of the deleting the vertex $W(T)$ (and edges\footnote{Here $u_b, u_g$ and $u_r$ denote the vertices which are linked to the writer by edges of colors blue, green and red, respectively.} $\{u_b , W(T)\} \in E_b(T), \{u_g , W(T)\} \in E_g(T), \{u_r , W(T)\} \in E_r(T) \}$ that were previously incident upon this vertex), and then adding three new vertices, $W(T)_b, W(T)_g, W(T)_r$, together with new blue edges $\{u_b, W(T)_b\},\{W(T)_g, W(T)_r\}$, and new green edges $\{u_g, W(T)_g\},\{W(T)_b, W(T)_r\}$, and new red edges $\{u_r, W(T)_r\},\{W(T)_b, W(T)_g\}$. The net effect of these substitutions is to replace $W(T)$ with a triangle upon vertices $W(T)_b, W(T)_g, W(T)_r$.

We can show $H^*$ is equivalent to $H(T+1)$ by explicitly constructing an isomorphism $f$ from $H^*$ to $H(T+1)$ which preserves the position of the writer and the colors of the edges. Here $f$ is a bijection from the vertex set of $H^*$ to the vertex set of $H(T+1)$ which is defined as follows:
\begin{itemize}
\item For each vertex $x$ of $H^*$ such that $x< W(T) = 4.t(T)+1$, we have $f(x) = x$.

\item For each vertex $x$ of $H^*$ such that $x> W(T) = 4.t(T)+1$, we have $f(x) = x +2$.

\item $f(W(T)_b) = 4.t(T) +1, f(W(T)_g) = 4.t(T) +2$ and $f(W(T)_r) = 4.t(T) +3$.
\end{itemize}

\subsection{Proof of Theorem \ref{goldentheorem}}

Let $(X^t, G^t)$ denote the state of the system on time step $t$, evolving under the rules shown in Figure \ref{goldenrule} (starting with $G^0$ as the cube).

\begin{figure}
\centerline{\includegraphics[scale=0.55]{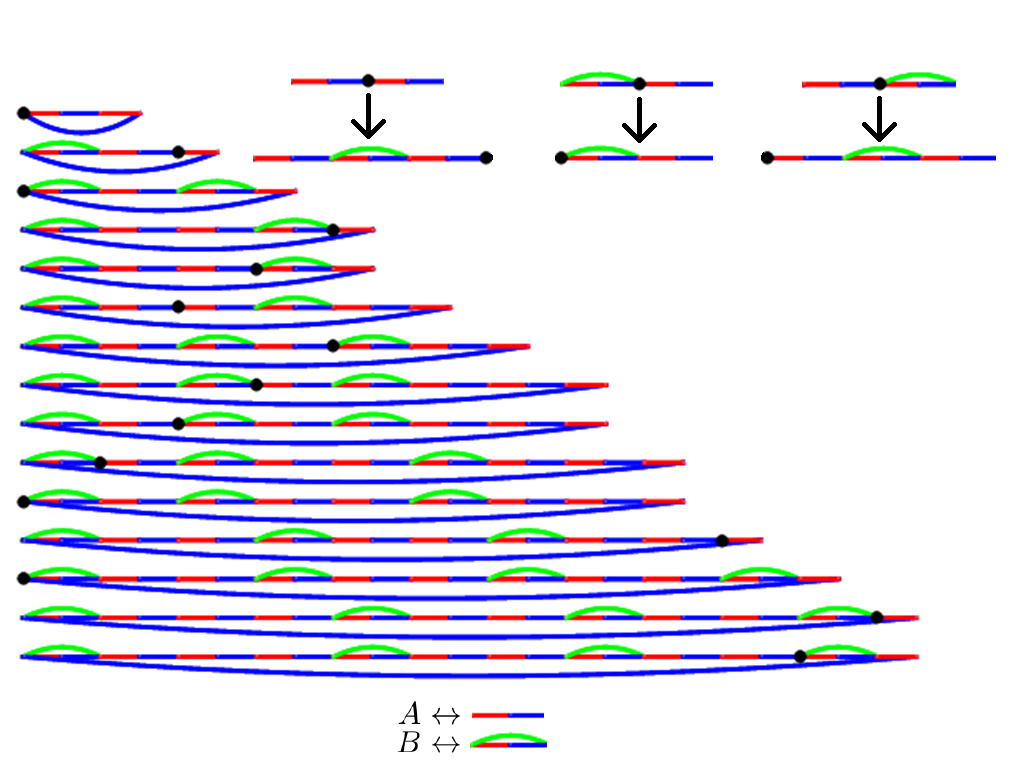}}
\vspace*{8pt}
\caption{
A representation of the (essentially one dimensional) dynamics of the system illustrated in Figure \ref{goldenrule}. This picture represents how the relevant part of the network changes over successive time steps, like Figure \ref{cycsimp1d}. The effective rewrite rules are shown at the top right. At the bottom we illustrate the symbolic substitutions we shall use to discuss this system algebraically.}
\label{golden1dB}
\end{figure}

Our method of proof will be to convert the system into a one dimensional rewrite system, and then derive a formula for the type $\rho(X^t)_{G^t}$ of the writer vertex at time $t$. We can use this information to characterize the number of vertices in $G^t$, because the number of vertices in $G^t$ only depends on the type the writer had on previous time steps (i.e. how much the network grew on previous time steps).
The way this system can be converted to a one dimensional rewrite system is illustrated in Figure \ref{golden1dB}. Rows of this figure (reading downwards) show the relevant part\footnote{The relevant part of the network $G^t$ is the subgraph induced on the cycle with edges of colors alternating between red and blue, which contains the writer, with green edges that are not part of triangles removed (see Figure \ref{cycsimp1d}).} of the network on subsequent time steps.

Let $A$ and $B$ denote the two kinds of vertices illustrated at the bottom of Figure \ref{golden1dB}. We can define these more rigorously, by saying that a vertex $v$ in $G^t$ of kind $A$ is a vertex which is not part of a triangle, that can be expressed as $v = X^t[r(br)^k]$, for some $k \geq 0$ and a vertex $v$ of kind $B$ is a vertex which is part of a triangle, that can be expressed as $v = X^t[r(br)^k]$, for some $k \geq 0$. Now we can represent our system as a word of symbols that evolves like a tag system \cite{nks}. In particular, the word $s(t) \in \{A,B\}^*$ gives the sequence of kinds of vertices one visits by starting at $X^t[r]$ (the vertex linked to the writer by a red edge) and repeatedly doing $br$ hops.
More precisely $s(t)=s(t)_0s(t)_1...s(t)_{|s(t)| -1}$ and $\forall k \in \{0,1,...,|s(t)| -1 \}$ we have
 \begin{eqnarray}
s(t)_k = \left\{\begin{array}{ll}
A & \text{if } \rho(X^t[r(br)^k])_{G^t}=0,\\
B  & \text{if } \rho(X^t[r(br)^k])_{G^t} = g.
\end{array}\right. ~
\end{eqnarray}
Here the length of our word is $|s(t)| = \min \{ k>0 : X^t[(rb)^k] = X^t \}$, which is the number of red edges in the alternating $rb$ cycle.

As the system is updated, the writer modifies the alternating $rb$ cycle and moves around it. The way the system gets updated is determined by the type of the writer, which is determined by the forms of $s(t)_{|s(t)|-1}$ and $s(t)_0$ (which represent the kinds of the vertices to the writer's left and right, respectively). In our representation, using $s(t)$, an update corresponds to making $s(t+1)$ by taking $s(t)$, altering the elements at the extreme left or right hand side (since these are the parts adjacent to the writer) and also rotating the word (to represent the writers movement).

$s(0)=AA$

$s(1) = ABA$

$s(2) = BABA$

$s(3) = ABAAB$

To see the way $s(t)$ evolves with time, note that $\forall t \geq 0$ we have:
\begin{enumerate}
\item If $s(t)=AxA$, for some $x \in \{A,B \}^*$ then $\rho(X^t)_{G^t}=0$ and $s(t+1)=xABA$.
\item If $s(t)=AxB$, for some $x \in \{A,B \}^*$ then $\rho(X^t)_{G^t}=r$ and $s(t+1)=BAx$.
\item If $s(t)=BxA$, for some $x \in \{A,B \}^*$ then $\rho(X^t)_{G^t}=b$ and $s(t+1)=ABAx$.
\end{enumerate}
These three statements are sufficient to describe the evolution of the system. They correspond to the rewrite rules shown at the top right of Figure \ref{golden1dB}. For a word $x = x_0 x_1...x_{m-1}$ let $\last(x) = x_{m-1}$ denote the last element of $x$. The key to our proof is to establish a correspondence between the dynamics of $s(t)$ and the dynamics of the binary tag system $K(T)$. Here $K(T)$ denotes the word obtained by updating $K(0)=0$, a number of times equal to $T$, under the tag system with rules $0 \rightarrow 01$ and $1 \rightarrow 011$. More precisely, $\forall T \geq 0$ we have:
\begin{enumerate}
\item If there exists a word $x \in \{ 0 , 1 \}^*$ such that $K(T) = x0$ then $K(T+1) = 01x$.
\item If there exists a word $x \in \{ 0 , 1 \}^*$ such that $K(T) = x1$ then $K(T+1) = 011x$.
\end{enumerate}

$K(0)=0$

$K(1)=01$

$K(2)=0110$

$K(3)=01011$

The correspondence between $K(T)$ and $s(t)$ involves the mapping $\mu: \{0,1 \}^* \mapsto \{A,B \}^*$ defined such that for each binary word $x \in \{0,1 \}^*$ we have that $\mu(x)$ is the word obtained by taking $x$ and replacing each $0$ with $AB$ and each $1$ with $A$. The correspondence is not trivial. When $\last(K(T)) = 0$, the single step $K(T) \rightarrow K(T+1)$ corresponds to two steps, $s(t) \rightarrow s(t+1) \rightarrow s(t+2)$, in the $s$ system. When $\last(K(T)) = 1$, the single step $K(T) \rightarrow K(T+1)$ corresponds to four steps, $s(t) \rightarrow s(t+1) \rightarrow s(t+2)\rightarrow s(t+3) \rightarrow s(t+4)$, in the $s$ system.

$\forall T \geq 0$ let us define the $s$-time $\gamma(T)$ corresponding with $T$ as
\begin{equation}\gamma(T) = 2 + 2.T + 2.|\{ i \in \{0,..,T-1\}: \last(K(i))=1 \}|.\end{equation}
Here $s(\gamma(T))$ represents $K(T)$, as we describe with our next result.
\begin{lemma}\label{coro}

$\forall T \geq 0$ the following equations hold:
\begin{equation}s(\gamma(T)) = B \mu(K(T)) A \end{equation}
\begin{equation}s(\gamma(T)+1) = ABA \mu(K(T)).\end{equation}
Moreover, if $\last(K(T))=1$ then the following two equations also hold:
\begin{equation}s(\gamma(T)+2) = BA \mu(K(T)) BA \end{equation}
\begin{equation}s(\gamma(T)+3) = ABAA\mu(K(T))B.\end{equation}
\end{lemma}

\begin{proof}
For $T=0$ we have $\gamma(T)=2$, $K(T) = 0$, $\last(K(T)) = 0$, $\mu(K(T))=AB$, $s(\gamma(T)) = BABA$ and $s(\gamma(T)+1) = ABAAB$, and so our result hold in this case.
Suppose the result holds for some $T \geq 0$. We will show that it holds for $T+1$.

First let us consider the case where $\last(K(T))=0$. In this case there must exist some $x \in \{ 0 , 1 \}^*$ such that $K(T) = x0$, and so we have
\begin{eqnarray}
\label{e1}
  s(\gamma(T)+1) & = & ABA \mu(K(T)) \\
   & = & ABA \mu(x) \mu(0) \\
   & = & ABA \mu(x) AB.
\end{eqnarray}
Here, Eq. (\ref{e1}) follows from our assumption that the result holds for $T$.
It follows (from the update rules of $s$) that $s(\gamma(T)+2) = BABA \mu(x) A$. Also, in this case,
\begin{eqnarray}
  \gamma(T+1) & = & 2 + 2.(T+1) + 2.|\{ i \in \{0,..,T\}: \last(K(i))=1 \}| \\
   & = & \gamma(T)+2.
\end{eqnarray}
Moreover, $K(T+1) = 01x$, and so we have
\begin{eqnarray}
  B \mu(K(T+1)) A & = & BABA\mu(x)A \\
   & = & s(\gamma(T+1)),
   \label{rat1}
\end{eqnarray}
(as we require). It follows (from the update rules of $s$) that
\begin{equation} s(\gamma(T+1)+1) = ABA\mu(K(T+1)) \label{rat2} \end{equation}
(as we require). Moreover, if $\last(K(T+1)) = 1$ then there exists a binary word $y \in \{0,1\}^*$ such that $K(T+1) = y1$. Now
\begin{eqnarray}
  s(\gamma(T+1)+1) & = & ABA\mu(y) \mu(1) \\
   & = & ABA\mu(y)A.
\end{eqnarray}
Now it follows (from the update rules of $s$) that
\begin{eqnarray}
  s(\gamma(T+1)+2) & = & BA\mu(y)ABA \\
   & = & BA\mu(y1)BA \\
   & = & BA\mu(K(T+1))BA,
   \label{rat3}
\end{eqnarray}
(as we require). Now it follows (from the update rules of $s$) that
\begin{equation} s(\gamma(T+1)+3) = ABAA\mu(K(T+1))B \label{rat4} \end{equation}
 (as we require). So now we have shown that the result holds for $T+1$ in this case (where $\last(K(T))=0$).

Now let us consider the other case, where $\last(K(T))=1$. In this case there must exist some $x \in \{ 0 , 1 \}^*$ such that $K(T) = x1$.
Our assumption that the result holds for $T$ gives us $s(\gamma(T)+3) = ABAA\mu(K(T))B$. It follows (from the update rules of $s$) that
\begin{eqnarray}
  s(\gamma(T)+4) & = & BABAA\mu(K(T)) \\
   & = & BABAA\mu(x)\mu(1) \\
   & = & BABAA\mu(x)A.
\end{eqnarray}
In this case,
\begin{eqnarray}
  \gamma(T+1) & = & 2 + 2.(T+1) + 2.|\{ i \in \{0,..,T\}: \last(K(i))=1 \}| \\
   & = & BABAA\mu(x)\mu(1) \\
   & = & \gamma(T)+4.
\end{eqnarray}
Hence $s(\gamma(T+1)) = s(\gamma(T)+4) = BABAA\mu(x)A$. Moreover, $K(T+1) = 011x$, and so we have
\begin{eqnarray}
  B \mu(K(T+1)) A & = & BABAA\mu(x)A \\
   & = & BABAA\mu(x)\mu(1) \\
   & = & s(\gamma(T+1)),
\end{eqnarray}
(as we require).
Now we can establish that
\begin{equation} s(\gamma(T+1)+1) = ABA\mu(K(T+1)), \end{equation}
in exactly the same way we showed Eq. (\ref{rat2}). In addition, if $\last(K(T+1)) = 1$ then there exists a binary word $y \in \{0,1\}^*$ such that $K(T+1) = y1$, and we can establish that
\begin{equation} s(\gamma(T+1)+2)  =  BA\mu(K(T+1))BA \end{equation}
and
\begin{equation} s(\gamma(T+1)+3) = ABAA\mu(K(T+1))B, \end{equation}
in exactly the same way as we showed Eq. (\ref{rat3}) and Eq. (\ref{rat4}), respectively.
So now we have shown that the result holds for $T+1$ in this case (where $\last(K(T))=1$). So now we have shown that, if the result holds for $T$ then it holds for $T+1$. Now since the result holds when $T=0$ we can use induction with $T$ to prove the result holds $\forall T \geq 0$.
\end{proof}

We can use this lemma to form the following corollary, that describes the sequence $\rho(X^0)_{G^0}\rho(X^1)_{G^1}...$ in terms of the forms of $K(T):T \geq 0$.
\begin{corollary}\label{corocoro}

$\forall t \geq 0$ we have the following:
\begin{enumerate}
\item If $t>0$ and $t$ is even then $\rho(X^t)_{G^t} = b$.
\item If $t \in \{0,1 \}$ or if $t$ is of the form
$$3 + 2.T + 2.|\{ i \in \{0,..,T-1\}: \last(K(i))=1 \}|,$$
for some $T \geq 0$, that is such that $\last(K(T))=1$ then $\rho(X^t)_{G^t} = 0$.
\item Otherwise\footnote{i.e., when $t>1$, $t$ is odd, and there does not exist any $T \geq 0$ such that $3 + 2.T + 2.|\{ i \in \{0,..,T-1\}: \last(K(i))=1 \}|$} we have $\rho(X^t)_{G^t} = r$.
\end{enumerate}
\end{corollary}

\begin{proof}
Clearly $\rho(X^0)_{G^0} = \rho(X^1)_{G^1} = 0$, so the result holds $\forall t \in \{0,1 \}$. For $t \geq 2$, let $T^*$ be the maximum integer such that $\gamma(T^*) \leq t$. We must have $T^* \geq 0$ since $\gamma(0)=2 \leq t$. Also, we must have $\gamma(T^*) \geq t-3$, since otherwise (i.e., if $\gamma(T^*) < t-3$) we would have $\gamma(T^*+1) \leq \gamma(T^*)+4 \leq t$ which would contradict our assumption that $T^*$ is maximum. Now we know $\gamma(T^*) \in \{t-3,t-2,t-1,t \}$ let us consider the different cases.
\ \\
If $\gamma(T^*) = t$ then $t$ is even and Lemma \ref{coro} gives us that $s(t) = B \mu(K(T^*)) A$ and so $\rho(X^t)_{G^t} =b$.
\ \\
If $\gamma(T^*) = t-1$ then $t = \gamma(T^*) +1$ is odd and Lemma \ref{coro} gives us that $s(t) = s(\gamma(T^*) +1) = ABA \mu(K(T^*))$ and so $\last(K(T^*)) = 0$ implies $\rho(G^t)_{X^t} = r$ and $\last(K(T^*)) = 1$ implies $\rho(G^t)_{X^t} = 0$.
\ \\
If $\gamma(T^*) = t-2$ then $t = \gamma(T^*) +2$ is even and Lemma \ref{coro} gives us that $s(t) = s(\psi(T^*) +2) = BA \mu(K(T^*)) BA$ and so $\rho(G^t)_{X^t} = b$.
\ \\
If $\gamma(T^*) = t-3$ then $t = \gamma(T^*) +3$ is odd and Lemma \ref{coro} gives us that $s(t) = s(\psi(T^*) +3) = ABAA\mu(K(T^*))B$ and so $\rho(G^t)_{X^t} = r$.

These are all the possibilities. So we have $\forall t \geq 2$ that, $\rho(G^t)_{X^t} = b$ if and only if $t$ is even. When $t \geq 2$ is odd, we must have $\rho(G^t)_{X^t} \in \{ 0 , r \}$.
\ \\
Moreover, $\rho(G^t)_{X^t} = 0$ if and only if $t = \gamma(T^*) +1$ where $\last(K(T^*)) = 1$. In other words, $\rho(G^t)_{X^t} = 0$ if and only if there exists some $T \geq 0$ such that $t = 1 + \gamma(T) = 3 + 2.T + 2.|\{ i \in \{0,..,T-1\}: \last(K(i))=1 \}|$ and $\last(K(T))=1$.
\ \\
Also, if $\rho(G^t)_{X^t} \notin \{ b , 0 \}$ then $\rho(G^t)_{X^t} = r$.
\end{proof}
Now we have established Corollary \ref{corocoro}, we can gain a description of the sequence $\rho(X^0)_{G^0}\rho(X^1)_{G^1}...$ by getting a formula for $\last(K(T))$. Our tag system $K(T)$ evolves like a slow version of the neighbor independent substitution system $Z(n)$, where every element of the word is updated every time step under rules $0 \rightarrow 01$, $1 \rightarrow 011$.

More precisely, let us define the mapping $\Omega: \{0,1\}^* \mapsto \{0,1\}^*$ such that $\forall x \in \{0,1\}^*$ we have that $\Omega(x)$ is the word obtained by taking $x$ and replacing each $0$ with $01$ and replacing each $1$ with $011$.

Now let $Z(0)=0$ and $\forall n \geq 0$ let $Z(n+1)=\Omega(Z(n))$


$Z(0)=0$

$Z(1)=01$

$Z(2)=01011$

$Z(3) = 0101101011011$.

For a word, $x = x_0x_1...x_{m-1}$ let $\REV(x) = x_{m-1}x_{m-2}...x_0$ denote the reversal of the order of characters in the word.
\begin{lemma}\label{revrev}

The infinite word $\last(K(0))\last(K(1))...$ is equal to the infinite word $\REV(Z(0))\REV(Z(1))...$
\end{lemma}
\begin{proof}
For a word $K(T) = K(T)_0K(T)_1...K(T)_{|K(T)|-1}$ we have
\begin{equation}K(T+1) = \Omega(K(T)_{|K(T)|-1})K(T)_0...K(T)_{|K(T)|-2}\end{equation} and
\begin{equation}K(T+2) = \Omega(K(T)_{|K(T)|-2})\Omega(K(T)_{|K(T)|-1})K(T)_0...K(T)_{|K(T)|-3}.\end{equation}
Following this pattern, one can see that
$$\last(K(T))\last(K(T+1))...\last(K(T+|K(T)|))$$
is equal to
$$K(T)_{|K(T)|-1} K(T)_{|K(T)|-2}...K(T)_0 = \REV(K(T)).$$
Moreover, after we have performed $|K(T)|$ updates on $K(T),$ we will get
\begin{equation}K(T+|K(T)|)=\Omega(K(T)_{0}) \Omega(K(T)_{1})...\Omega(K(T)_{|K(T)|-}) = \Omega(K(T)).\end{equation}
From this one can see that $K(T)$ evolves like a slow version of the substitution system $Z(n)$. When $K(T) = Z(n)$ (as is the case when $T=n=0$) we will have \begin{equation}K(T+|K(T)|) = \Omega(K(T)) =\Omega(Z(n)) = Z(n+1)\end{equation}
 and
 \begin{equation}\last(K(T))\last(K(T+1))...\last(K(T+|K(T)|)) = \REV(Z(n)).\end{equation}
 Our result can then be proved by induction.
 \end{proof}
\begin{lemma}\label{glue}

$\forall n \geq 1$ we have $Z(n+1)=Z(n)Z(n)Z(n-1)...Z(1)1$
\end{lemma}

\begin{proof}
$Z(2)=01011=Z(1)Z(1)1$ so the result holds for $n=1$. Suppose that the result holds for some $n \geq 1$, now
\begin{eqnarray}
  Z(n+2) & = & \Omega(Z(n+1)) \\
   & = & \Omega(Z(n)Z(n)Z(n-1)...Z(1)1) \\
   & = & Z(n+1)Z(n+1)Z(n)...Z(2)011 \\
   & = & Z(n+1)Z(n+1)Z(n)...Z(2)Z(1)1,
\end{eqnarray}
and so the the result holds for $n+1$. The result is proved, in general, by induction with $n$.
\end{proof}

For a word $x=x_0x_1...x_{m-2}x_{m-1} \in \{0,1 \}^*$ let $\tog(x)=(1-x_0)x_1...x_{m-2}(1-x_{m-1})$ be the word obtained by altering the characters at $x$'s extreme left and right hand sides.

Let $L(x) = x_1x_2...x_{m-1}$ denote the word obtained by deleting the left most element from $x$.

\begin{lemma}\label{tog}
$\forall m \geq 1$, for any sequence of words $a^1,a^2...,a^m \in \{01,011\}$ we have $\tog(a^1)\tog(a^2)...\tog(a^m)=L(a^1)a^2...a^m0$.

\end{lemma}

\begin{proof}
This clearly holds when $m=1$ since $\tog(01)=10=L(01)0$ and $\tog(011)=110 = L(011)0$. Suppose that the result holds for $m \geq 1$.
Now $\tog(a^1)\tog(a^2)...\tog(a^{m})\tog(a^{m+1})$
is equal to
$L(a^1)a^2...a^{m}0\tog(a^{m+1})$, which is equal to
$L(a^1)a^2...a^{m}0L(a^{m+1})0$, which is equal to $L(a^1)a^2...a^{m}a^{m+1}0$ as required. So the result holds for $m+1$. By induction with $m$ the result is proved in general.

\end{proof}

\begin{lemma}\label{togrev}

$\forall n \geq 1$ we have $\tog(Z(n)) = \REV(Z(n))$, $Z(n)_0 = 0$ and $\last(Z(n))=1$.

\end{lemma}

\begin{proof}

$\tog(Z(1))=10= \rev(Z(1))$, so our result holds when $n=1$. Suppose the result holds for some $n \geq 1$. We write $Z(n) = Z(n)_0...Z_{|Z(n)|-1}$, now
\begin{eqnarray}
  \tog(Z(n+1)) & = & \tog(\Omega(Z(n)_0)...\Omega(Z(n)_{|Z(n)|-1})) \\
   & = & \tog(01\Omega(Z(n)_1)...\Omega(Z(n)_{|Z(n)|-2})011) \\
   & = & 11\Omega(Z(n)_1)...\Omega(Z(n)_{|Z(n)|-2})010,
   \label{jerry}
\end{eqnarray}
also,
\begin{eqnarray}
  \REV(Z(n+1)) & = & \REV(\Omega(Z(n)_{|Z(n)|-1}))...\REV(\Omega(Z(n)_{0})) \\
   \label{tom0}
   & = & 110 \REV(\Omega(Z(n)_{|Z(n)|-2}))...\REV(\Omega(Z(n)_{1}))10 \\
   \label{tom1}
   & = & 110 \REV(\Omega(Z(n)_{1}))...\REV(\Omega(Z(n)_{|Z(n)|-2}))10 \\
   \label{tom2}
   & = & 110 \tog(\Omega(Z(n)_{1}))...\tog(\Omega(Z(n)_{|Z(n)|-2}))10 \\
    \label{tom3}
   & = & 110 L(\Omega(Z(n)_{1}))\Omega(Z(n)_{2})...\Omega(Z(n)_{|Z(n)|-2})010 \\
   \label{tom4}
    & = & 11 \Omega(Z(n)_{1})\Omega(Z(n)_{2}...\Omega(Z(n)_{|Z(n)|-2})010
\end{eqnarray}
The reason Eq. (\ref{tom0}) implies Eq. (\ref{tom1}) is that $\REV(Z(n)) = \tog(Z(n))$ implies
$$Z(n)_{1}...Z(n)_{|Z(n)|-2} = Z(n)_{|Z(n)|-2}...Z(n)_{1}$$
 is palindromic, which implies
 $$\REV(\Omega(Z(n)_{|Z(n)|-2}))...\REV(\Omega(Z(n)_{1})) = \REV(\Omega(Z(n)_{1}))...\REV(\Omega(Z(n)_{|Z(n)|-2})).$$
 The reason Eq. (\ref{tom1}) implies Eq. (\ref{tom2}) is that we always have $\Omega(Z(n)_{i}) \in \{01,011\}$, and so $\REV(\Omega(Z(n)_{i})) = \tog(\Omega(Z(n)_{i}))$. We can use Lemma \ref{tog} to derive Eq. (\ref{tom3}) from Eq. (\ref{tom2}), and Eq. (\ref{tom4}) follows because the left most entry in $\Omega(Z(n)_{1})$ is $0$.

Now, since Eq. (\ref{tom4}) is identical to Eq. (\ref{jerry}), we have proved that our result holds for $n+1$. The fact that the result holds $\forall n \geq 1$ then follows by induction.

\end{proof}

Now we can state our critical result. Recall that  $\last(K(0))\last(K(1))...$ is the infinite sequence consisting of the last elements of the words $K(0),K(1)...$ (evolving under our tag system).

\begin{lemma}\label{crit}
$\last(K(0))\last(K(1))...=\lim_{n \rightarrow \infty} Z(n)$
\end{lemma}

\begin{proof}

Lemma \ref{glue} gives us that $\forall n>0$,
\begin{equation} Z(n+1) = Z(n)Z(n)Z(n-1)...Z(1)1, \end{equation}
and so
\begin{equation} \REV(Z(n+1)) = 1 \REV(Z(1)) \REV(Z(2))...\REV(Z(n))\REV(Z(n)), \end{equation}
Now, Lemma \ref{togrev} implies $\REV(Z(k))=\tog(Z(k))$, $\forall k >0$. So now we have
\begin{eqnarray}
  Z(n+1) & = & \REV(\REV(Z(n+1))) \\
   & = & \tog( \REV( Z(n+1))) \\
   & = & 0 \REV(Z(1)) \REV(Z(2))...\REV(Z(n))R(\REV(Z(n)))1,
\end{eqnarray}
where $R(x)$ denotes the word obtained by deleting the right most element from $x$.
Now, taking the limit as $n \rightarrow \infty$ we have
\begin{eqnarray}
  \lim_{n \rightarrow \infty} (Z(n)) & = & \lim_{n \rightarrow \infty} (Z(n+1)) \\
   & = & 0 \REV(Z(1)) \REV(Z(2))... \\
   \label{thom1}
   & = & \REV(Z(0)) \REV(Z(1)) \REV(Z(2))... \\
   \label{thom2}
   & = & \last(K(0))\last(K(1))\last(K(2))...
   \end{eqnarray}
Here Lemma \ref{revrev} gives us that Eq. (\ref{thom1}) implies Eq. (\ref{thom2}).
\end{proof}







The above result (Lemma \ref{crit}) is important because a formula is known \cite{A096270} for the infinite sequence of $\lim_{n \rightarrow \infty} (Z(n))$ (which we have just shown to be equal to $\last(K(0))\last(K(1))...$). Using this formula we have $\forall i \geq 0$, that
\begin{equation} \last(K(i)) = \left(\lim_{n \rightarrow \infty} (Z(n)) \right)_i = \lfloor (i+1) \phi \rfloor - \lfloor i. \phi \rfloor -1,\end{equation}
where $\phi = \frac{1+ \sqrt{5}}{2}$ is the golden ratio. Other results are also known for this sequence.
For $T \geq 0$ the index of the $T$th $1$ in the sequence $\lim_{n \rightarrow \infty} (Z(n))$ is well known \cite{A000201} to be given by the formula $\lfloor (T+1) \phi \rfloor$. In other words,
\begin{equation}\left\{ T \geq 0 : \left(\lim_{n \rightarrow \infty} (Z(n))\right)_{T}=1 \right\} = \{ T \geq 0 : \last(K(T))= 1 \} = \{ \lfloor (n+1) \phi \rfloor : n \geq 0 \}\end{equation}

Suppose $h \geq 2$. Now, according to Corollary \ref{corocoro}, $\rho(X^h)_{G^h} = 0$ if and only if $\exists T \geq 0$ such that $h = 3 + 2 T + 2 |\{i \in \{0,1,...,T-1 \}: \last(K(i))=1 \}|$ and $\last(k(T)) = 1$.

Now since  $T \geq 0, \last(k(T)) = 1$ if and only if $\exists n \geq 0 : \lfloor (n+1) \phi \rfloor =T$ we have the following:
\ \\
\emph{$\rho(X^h)_{G^h} = 0$ if and only if $\exists n \geq 0$ such that $h =3 + 2\lfloor (n+1) \phi \rfloor + 2.\Theta$,} where:
\begin{eqnarray}
 \Theta & = & | \{ i \in \{ 0 , 1,...,\lfloor (n+1) \phi \rfloor - 1 \} : \last(K(i))=1 \}| \\
    & = & | ( \{ 0 , 1 ,..., \lfloor (n+1) \phi \rfloor - 1 \} \cap \{ \lfloor (0+1) \phi \rfloor, \lfloor (1+1) \phi \rfloor,... \})| \\
   & = & | \{  \lfloor (0+1) \phi \rfloor, \lfloor (1+1) \phi \rfloor,...,\lfloor ((n-1)+1) \phi \rfloor \}| \\
   & = & n.
   \end{eqnarray}


Hence we have
\begin{equation}\label{one}
\{h \geq 2 : \rho(X^h)_{G^h} = 0 \} = \{ 3 + 2\lfloor (n+1) \phi \rfloor + 2.n : n \geq 0\}.
\end{equation}
Now, the number of vertices $|G^t|$, in $G^t$ will be equal to
\begin{equation} |G^t| = 8 + 2 |\{ i \in \{0,...,t-1\}: \rho(X^i)_{G^i} =b  \}|+ 2|\{ i \in \{0,...,t-1\}: \rho(X^i)_{G^i} =0 \}| \label{impo} \end{equation}
because the number of vertices only changes when the writer is of type $0$ or $b$, and it increases by $2$ each such time. Now Corollary \ref{corocoro} implies \begin{equation} |\{ i \in \{0,...,t-1 \}: \rho(X^i)_{G^i} = b \}| = \left\lfloor \frac{t+1}{2} \right\rfloor -1 + \delta_{t,0} \label{subb} \end{equation} (where $\delta_{i,j}$ is the Kronecker Delta).

Also, $\forall t>5$, we have:
$$|\{ i \in \{0,...,t-1 \}: \rho(X^i)_{G^i} =0 \}|  =  2 + |\{ i \in \{2,3...,t-1 \}: \rho(X^i)_{G^i} =0 \}|$$
\begin{equation} \label{sub0a} \end{equation}
$$|\{ i \in \{0,...,t-1 \}: \rho(X^i)_{G^i} =0 \}|  =  2 + |\{ n \geq 0 : 3 + 2\lfloor (n+1) \phi \rfloor + 2.n <t \}|,$$
\begin{equation} \label{sub0b} \end{equation}
Where we use use  Eq. (\ref{one}) to infer Eq. (\ref{sub0b}) from Eq. (\ref{sub0a}).  Now $\forall n \geq 0$ $\forall t>5$ the following four inequalities are equivalent:
\begin{equation}\label{two}
3 + 2\left\lfloor (n+1) \phi \right\rfloor + 2.n <t
\end{equation}
\begin{equation}\label{three}
\left\lfloor (n+1) \phi \right\rfloor < \left\lfloor \frac{t}{2} \right\rfloor -(n+1)
\end{equation}
\begin{equation}\label{four}
(n+1)\phi^2 < \left\lfloor \frac{t}{2} \right\rfloor
\end{equation}
\begin{equation}\label{five}
n < \left\lceil \left\lfloor\frac{t}{2} \right\rfloor\frac{1}{\phi^2}\right\rceil-1
\end{equation}
The equivalence of inequalities (\ref{two}) and (\ref{three}) can be derived using the fact that $2a <b \Leftrightarrow a < \left\lfloor \frac{b+1}{2} \right\rfloor$, for integers $a$ and $b$. The equivalence of inequalities (\ref{three}) and (\ref{four}) can be derived using the fact that $\lfloor r \rfloor <m \Leftrightarrow r <m$ for a real $r$ and an integer $m$, and the fact that $1+ \phi = \phi^2$.

It follows that when $t>5$ we can rewrite Eq. (\ref{sub0b}) as
\begin{eqnarray}
 |\{ i \in \{0,...,t-1 \}: \rho(X^i)_{G^i} =0 \}| & = & 2+ \left|\left\{ n \geq 0 : n < \left\lceil \left\lfloor\frac{t}{2} \right\rfloor\frac{1}{\phi^2}\right\rceil-1 \right\}\right| \\
    & = & 2+ \left|\left\{ 0,1,...,\left\lceil \left\lfloor\frac{t}{2} \right\rfloor\frac{1}{\phi^2}\right\rceil-2 \right\}\right| \\
    \label{sub0c}
    & = & 2+ \left\lceil \left\lfloor\frac{t}{2} \right\rfloor\frac{1}{\phi^2}\right\rceil-1
    \end{eqnarray}
Now by substituting Eq. (\ref{subb}) and Eq. (\ref{sub0c}) into Eq. (\ref{impo}), we get that when $t>5$, the number of vertices in $G_t$ is
\begin{equation}|G^t| = 8 + 2 \left(\left\lfloor \frac{t+1}{2} \right\rfloor -1 \right) + 2 \left( 2 + \left\lceil \left\lfloor\frac{t}{2} \right\rfloor\frac{1}{\phi^2}\right\rceil-1 \right), \label{finall} \end{equation}
as required. It can also be checked that Eq. (\ref{finall}) gives the correct answer for each $t \in \{0,1,2,3,4,5 \}$. Q.E.D.

\ \\

\end{document}